\documentclass[prd,preprint,preprintnumbers,floats,floatfix,apsrev,amsmath,amsfonts,nofootinbib,superscriptaddress]{revtex4-2}
\usepackage{color}
\usepackage{graphicx}
\usepackage{dcolumn}
\usepackage{natbib}
\usepackage{hyperref}
\usepackage{hypernat}
\usepackage{natbib}
\usepackage{soul}
\usepackage{amssymb}
\usepackage{xcolor}
\newcommand{\dd}{{\rm d}}

\newcommand{\p}{\partial}

\newcommand{\ret}{\mathrm{ret}} 
\newcommand{\be}{\begin{equation}}
\newcommand{\en}{\end{equation}} 
\newcommand{\riso}{\bar{{\rm r } } }

\newcommand{\aiso}{\bar{{\rm a } } }
\newcommand{\biso}{\bar{{\rm b}} }

\newcommand{\paren}[1]{\left({#1}\right)}
\newcommand{\rtil}{\tilde{r}}
\newcommand{\atil}{\tilde{\alpha}}

 \usepackage{amsthm}

 \newtheorem{theorem}{Theorem}

\begin{document}
\title{On the radiation field of a linearly accelerated charged particle in Born-Infeld theory}

\author{Miguel~L.~Pe\~{n}afiel}
\email{mpenafiel@cbpf.br}
\email{miguelpenafiel@upb.edu}
\affiliation{UERJ - Universidade do Estado do Rio de Janeiro 150, CEP 20550-013, Rio de Janeiro, RJ, Brazil.}
\affiliation{FIA - Facultad de Ingeniería y Arquitectura, Universidad Privada Boliviana, Camino Achocalla Km 3.5, La Paz, Bolivia.}

\author{Santiago Esteban Perez Bergliaffa}
\email{santiagobergliaffa@uerj.br}
\affiliation{UERJ - Universidade do Estado do Rio de Janeiro 150, CEP 20550-013, Rio de Janeiro, RJ, Brazil.}

\date{\today}

\begin{abstract}
The electric potential and the electromagnetic field for a linearly accelerated Born-Infeld charged particle 
are obtained 
in an inertial frame
by a method that can, in principle, be  applied to any electromagnetic theory. The method is based on \textit{(i)}
the fact that the metric near the horizon of a Schwarzschild black hole is equivalent to that of Rindler spacetime, and \textit{(ii)} a theorem that guarantees that the electrostatic potential for a given nonlinear theory in a static, spherically symmetric spacetime is entirely specified by 
the Maxwellian electrostatic potential in the same background. Using analytical and numerical methods, the features of the radiation field and the radiation-reaction for such an accelerated particle are discussed in detail.
\end{abstract}

\maketitle

\section{Introduction} \label{sec:intro}

The study of the radiation emitted by an accelerated charged particle is usually performed in Maxwell Electrodynamics (ME) through the electromagnetic potential for the particle, known as the Li\'{e}nard-Wiechert (LW) potential, the derivation of which relies on the well-known Green's function method. While
the properties of the radiation field are obtained in a rather straightforward way, the description of the action of the field on the particle within ME has been theoretically challenging.
Divergences arise, along with runaway solutions \cite{Dirac1938}. Although most of these theoretical inconsistencies were overcome with the advent of Quantum Electrodynamics (QED),  the study of such classical problems remains of interest  today \cite{Gralla2009,Kiessling2019,Ruhlandt2020}.

The situation changes drastically when
the same problem is considered in nonlinear electrodynamics (NLED). In particular, the Green's function method is no longer applicable \cite{Perlick2015}. The radiation field has not been studied, to our knowledge, in NLED with the exception of ModMax theory \cite{Lechner2022}.
As we shall see, the well-known equivalence between the geometry
near the horizon of a Schwarzschild black hole and that of Rindler's spacetime paves the way to the study of such problems\footnote{
The electromagnetic field of a uniformly accelerated charge in its comoving
Rindler frame was studied in \cite{Eriksen2004}, and the class of observers that detect an electrostatic field for a linearly accelerating charge have been found in  \cite{Leonov2012}.}.

Let us first point out that a closed analytic form for the potential of a static electric charge in the vicinity of a Schwarzschild black hole was obtained by Copson \cite{Copson1928}, and later modified to satisfy the adequate asymptotic boundary conditions by Linet \cite{Linet1976}\footnote{Subsequently, 
{other interesting phenomena in Maxwell's electrostatics have been addressed in other types of spacetimes} \cite{Leaute1976,Linet2005,Watanabe2013,Boisseau2013,RubindeCelis2020}.}. The analogous problem in the context of Born-Infeld (BI) electrodynamics has been recently studied in \cite{Falciano2019}, where an approximate expression for the electrostatic potential in good agreement with the numerical results was found\footnote{The BI particle in an Einstein-Born-Infeld spacetime was analyzed in \cite{Falciano2021}.}.

It is important to remark that, among the theories of NLED \cite{Plebanski1970,Boillat1970,OliveiraCosta2009}, Born-Infeld (BI) electrodynamics \cite{Born425} stands as one of the most interesting examples (see for instance \cite{Fradkin1985,Carley2006,Ferraro2007,Sanchez2018}). BI electrodynamics was specifically designed to eliminate the divergences that arise in the classical theory of point electrons.
The theory has a maximum field parameter, $\beta$ which, as we shall see, defines the size of the region where nonlinearities dominate. It is important to point out that
the nonlinearities of BI electrodynamics are often difficult to handle. In fact, only a handful of analytic solutions are known in the literature \cite{Ferraro2007,Ferraro2013,Manojlovic2020,Chruscinski1998a,Chruscinski1998}. 

Here we shall use the results presented in \cite{Falciano2019}, along with the abovementioned equivalence, to obtain the electromagnetic field of a linearly accelerated charged particle in BI theory. In brief, the field  of a charged BI particle at rest in an accelerated coordinate system will be obtained from the electric potential of a charged BI particle near the horizon of a Schwarzschild black hole, and then transformed to an inertial frame, a procedure analogous to that used in \cite{Gupta1998} to obtain the electromagnetic field in ME\footnote{It is also worth noting that fields in BI electrodynamics seen by accelerating observers can have interesting physical consequences \cite{Guzman‐Herrera2022}.}.
With the expression for the field at hand, we shall discuss the radiation field of a BI particle with constant acceleration in linear motion, and the corresponding radiation-reaction.

The paper is organized as follows. In Section \ref{sec:NLEDpoint} we review  some relevant
results of BI
electrodynamics for point particles, with emphasis on the electrostatic potential of a BI charged particle in the vicinity of a Schwarzschild black hole
obtained in \cite{Falciano2019}. In Section \ref{sec:Rindler1}, such results are adapted to a BI particle arbitrarily close to the horizon of a Schwarzschild black hole and the limit where the physical situation is that of a linearly accelerating particle in Rindler spacetime is discussed. Additionally, we derive the relevant components for the electric field in such a system. The explicit expressions for the fields transformed to an inertial coordinate system 
are derived in 
Section \ref{sec:inertcomp}. The phenomena of radiation and radiation-reaction for the linearly accelerated BI particle are investigated in 
Section \ref{sec:Radiation}, using the expressions for the fields previously derived. Finally, in Section \ref{sec:Conclusions}, we conclude by discussing the relevant results and suggesting further avenues for research. 

In the present work we use geometrized units such that $c=G=1$ and adopt Gaussian units for the electromagnetic field. Multiple coordinate systems, both inertial and non-inertial, 
are used in the text. 
They are summarized in Table \ref{table}.
\begin{table} \label{table}
    \centering
    \begin{tabular}{c|c|c}
        \textbf{Coordinate system} & \textbf{Coordinates} & \\ \hline
         flat coordinates& $(t',z',\theta,\varphi)$ & inertial\\
        $\lambda$-Rindler coordinates & $(\tau,\lambda,\rho,\varphi)$ & non-inertial \\
         $Z$-Rindler coordinates& $(\tau,Z,\rho,\varphi)$& non-inertial\\
    \end{tabular}
    \caption{Relevant coordinate systems used throughout the paper.}
\end{table}

\section{Nonlinear electrodynamics for point particles} \label{sec:NLEDpoint}

Theories of nonlinear electrodynamics (NLED) are
generalizations of Maxwell's theory that aim at solving
theoretical inconsistencies of the latter, 
as well as describing fundamental interactions at the effective level \cite{Plebanski1970,Boillat1970,OliveiraCosta2009,Ejlli2020,Sorokin2021}. Maxwell's equations are linear in the fields and hence follow from a variational principle applied to quadratic combinations of the fields, \textit{i.e.}, $\mathcal{L}_\mathrm{M}=-F/8\pi$, where $F=\frac{1}{2}F_{\mu\nu}F^{\mu\nu}$ is the invariant constructed from the Faraday tensor $F_{\mu\nu}=\p_\mu A_\nu-\p_\nu A_\mu$. From a theoretical standpoint, Lagrangians involving higher-order derivatives of the fields \cite{Podolsky1942} or even nonlinear functions on the two linearly independent invariants $F$ and $G$ can be considered, where $G=\frac{1}{2}\widetilde{F}^{\mu\nu}F_{\mu\nu}$, 
$\widetilde{F}^{\mu\nu}=\frac{1}{2}\eta^{\mu\nu\alpha\beta}F_{\alpha\beta}$, and $\eta^{\mu\nu\alpha\beta}$ is the totally antisymmetric Levi-Civita tensor. The field equations that follow from an arbitrary Lagrangian will depend on its specific functional form, which can be constrained by imposing physical requirements, such as duality and Maxwell's weak-field limit \cite{Plebanski1970,SalazarI.1987,SalazarIbarguen1989,GarciaD.1984,OliveiraCosta2009}.

Adopting the usual coupling between fields and spinless particles, the action is given by \be\label{eq:action}
S=-\int\dd s\paren{m\sqrt{g_{\mu\nu}\dot{x}^\mu\dot{x}^{\nu}}+eA_{\mu}\dot{x}^\mu}+\frac{1}{4\pi}\int{\dd^4x\mathcal{L}\paren{F,G}}
\en
where $\dot{x}^\mu$ stands for the four-velocity of the particle, $A_\mu$ is the four-potential associated with the electromagnetic field, $\mathcal{L}\paren{F,G}$ is an arbitrary Lagrangian density, $m$ stands for the bare mass of the particle, and $e$ denotes its electric charge. The equations of motion that follow from Eq. \eqref{eq:action} are given by
\begin{eqnarray}
f^\mu=m a^\mu&=&e F^{\mu}_{\ \nu}\frac{\dd x^\nu}{\dd s} \ ,\label{eq:Lorentz}\\ 
\p_\mu E^{\mu\nu}&=&-4\pi j^{\nu}\ , \label{eq:exeqn}\\
\p_\mu\widetilde{F}^{\mu\nu}&=&0\ , \label{eq:Bianchi}\\
j^{\mu}\paren{x}&=&e\int_{-\infty}^{\infty} \dd s v^{\mu}(s)\delta^{4}\left[x-z(s)\right]\ .\label{eq:current}
\end{eqnarray}
Here, Eq. \eqref{eq:Lorentz} is the usual expression for the Lorentz force, and the electromagnetic field equation in Eq. \eqref{eq:exeqn} is expressed in terms of the excitation tensor { $E_{\mu\nu}$,} defined as
{\be \label{eq:defexcitation}
E_{\mu\nu}=\frac{\p\mathcal{L}}{\p F^{\mu\nu}}=2\paren{\mathcal{L}_FF_{\mu\nu}+\mathcal{L}_G\widetilde{F}_{\mu\nu}}\ ,
\en
where $\mathcal{L}_X\equiv\p\mathcal{L}/\p X$}. Note that, by construction, $E^{\mu\nu}$ is antisymmetric, and Eq. \eqref{eq:exeqn} closely resembles the inhomogeneous Maxwell's equations. The key difference is that $E^{\mu\nu}$ is, in general, a nonlinear function of $F^{\mu\nu}$, which can be interpreted as a constituitive relation for the vacuum, leading to effects such as birrefringence \cite{Plebanski1970,Boillat1970,Boillat1966}. Additionally, Eq. \eqref{eq:Bianchi} is the Bianchi identity and Eq. \eqref{eq:current} gives the definition of the four-current of a point particle moving along and arbitrary worldline $z^{\mu}\paren{s}$. 

In this work, we consider a well-known example of NLED proposed by Born and Infeld (BI) \cite{Born425} and 
defined by the Lagrangian density
\be \label{eq:LBI}
\mathcal{L}_{\text{BI}}=\frac{1}{4\pi}\beta^2\paren{1-\sqrt{U}}\ , 
\en
where $U\equiv1+F/\beta^2-G^2/4\beta^4$ and $\beta$ is the maximum field parameter. 

The field equations can be written in the usual vector notation as
\begin{subequations}\label{eq:MaxBI}
\begin{eqnarray} 
\boldsymbol{\nabla}\cdot\mathbf{D}=4\pi\rho\ , \boldsymbol{\nabla}\times\mathbf{H}=\mathbf{J}+\frac{\p \mathbf{D}}{\p t}\ , \\
\boldsymbol{\nabla}\times\mathbf{E}=-\frac{\p\mathbf{B}}{\p t}\ , \ \boldsymbol{\nabla}\cdot\mathbf{B}=0\ ,
\end{eqnarray}
\end{subequations}
where the
spacelike vector fields 
$D^{\mu}\equiv -E^{\mu\nu}v_\nu$ and $H^\mu\equiv-\widetilde{E}^{\mu\nu}v_{\nu}$ were introduced, with $v^\mu$ being the normalized 4-velocity of the observer's worldline. Both fields are nonlinearly dependent of the electric and magnetic induction fields, constructed as the irreducible parts of the Faraday tensor $F_{\mu\nu}$. 

Additionally, BI electrodynamics possesses the so-called $F-P$ duality \cite{Bronnikov2001} 
which ensures that there exists a one-to-one correspondence between two different possible representations of the Lagrangian density via a Legendre transform. For a generic NLED, it is possible to define a Hamiltonian-like density, $\mathcal{H}=\mathcal{H}(P,S)$ from which the field equations can be derived. Both Hamiltonian-like and Lagragian pictures are connected via a Legendre transform:
\be
\mathcal{H}=E^{\mu\nu}F_{\mu\nu}-\mathcal{L}\ .
\en
In order to complete the Legendre transformation, it must be possible to invert the excitation tensor such that $E^{\mu\nu}\paren{F^{\mu\nu}}\to F^{\mu\nu}\paren{E^{\mu\nu}}$, and the Hamilton equations of motion are completed with
\be\label{eq:Hframe}
F^{\mu\nu}=2\frac{\p\mathcal{H}}{\p E_{\mu\nu}}=2\paren{\mathcal{H}_PE^{\mu\nu}+\mathcal{H}_S\widetilde{E}^{\mu\nu}}\ ,
\en
where $\mathcal{H}_X\equiv\p\mathcal{H}/\p X$. Therefore the quantity $\mathcal{H}$ can be expressed as a function of only the two invariants constructed from $E^{\mu\nu}$ and its dual, \textit{i.e.,} $\mathcal{H}=\mathcal{H}\paren{P,S}$ with
\begin{align}
P=\frac{1}{2}E^{\mu\nu}E_{\mu\nu}\ , \\
S=\frac{1}{2}\widetilde{E}^{\mu\nu}E_{\mu\nu}\ .
\end{align}
 For the case of BI electrodynamics, the following relations are valid:
\begin{subequations}
\begin{eqnarray}
E^{\mu\nu}&=&-\frac{1}{\sqrt{U}}\paren{F^{\mu\nu}-\frac{G}{2\beta^2}\widetilde{F}^{\mu\nu}}\ , \\
F^{\mu\nu}&=&-\frac{1}{\sqrt{V}}\paren{E^{\mu\nu}+\frac{S}{2\beta^2}\widetilde{E}^{\mu\nu}}\ \label{eq:FrelE},
\end{eqnarray}
\end{subequations}
where $V\equiv1-P/\beta^2-S^2/4\beta^4$
.
Using Eq. \eqref{eq:FrelE} and its dual, we find that the electric and magnetic fields are given by
\begin{eqnarray} \label{eq:EBvectors}
\mathbf{E}\equiv\frac{1}{\sqrt{V}}\paren{\mathbf{D}+\frac{\paren{\mathbf{D}\cdot\mathbf{H}}}{\beta^2}\mathbf{H}}\ , \\
\mathbf{B}\equiv\frac{1}{\sqrt{V}}\paren{\mathbf{H}-\frac{\paren{\mathbf{D}\cdot\mathbf{H}}}{\beta^2}\mathbf{D}}\ . \nonumber
\end{eqnarray}
In the electrostatic case, 
\be \label{eq:normEBI}
\mathbf{E}=\frac{\mathbf{D}}{\sqrt{1+\vert\mathbf{D}\vert^2/\beta^2}}\ .
\en

The field equations in Eq. \eqref{eq:MaxBI} indicate that the sources generate the fields $\mathbf{D}$ and $\mathbf{H}$, and these fields yield regular solutions for the fields $\mathbf{E}$ and $\mathbf{B}$ via the constitutive relations given by Eq. \eqref{eq:EBvectors}. This can be straightforwardly seen in the case of a point particle with charge $e$, where $\mathbf{D}=e/r^2\hat{\mathbf{r}}$. Defining the length parameter $\lambda_{\text{BI}}\equiv\sqrt{e/\beta}$, by Eq. \eqref{eq:normEBI}, the electric field is
\be \label{eq:EfieldBI}
\mathbf{E}=\frac{\beta}{\sqrt{1+r^4/\lambda_{\text{BI}}^4}}\hat{\mathbf{r}}\ ,
\en
which, when evaluated at the origin, gives the finite value $\vert\mathbf{E}\vert=\beta$. The electrostatic potential for a point particle in flat spacetime follows from the integration of the electric field along a 
trajectory, 
and is finite at the position of the particle, namely
\be\label{eq:potBIfst}
\phi_{\text{BI}}=-\int_{\infty}^{0}\mathbf{E}\cdot\dd\mathbf{r}=\sqrt{e\beta}\frac{\Gamma\paren{\frac{1}{4}}^2}{4\sqrt{\pi}}\ .
\en
Finally, the energy-momentum tensor for Born-Infeld electrodynamics is obtained through the definition $T_{\mu\nu}=\frac{2}{\sqrt{-g}}\frac{\delta \sqrt{-g}\mathcal{L}}{\delta g^{\mu\nu}}$, resulting in 
\be \label{eq:Tmn}
T_{\mu\nu}=\frac{1}{4\pi}\left[\frac{1}{\sqrt{U}}\left(F_{\nu}^{\ \lambda}F_{\lambda\mu}+\frac{1}{2\beta^2}G^2g_{\mu\nu}\right)+g_{\mu\nu}\beta^2\left(\sqrt{U}-1\right)\right]\ .
\en

\subsection{Born-Infeld particle in the vicinity of a Schwarzschild black hole}\label{ssec:BIpart}
In this subsection, we summarize the results of \cite{Falciano2019}, where the electrostatic potential of a BI particle in the vicinity of a Schwarzschild black hole was obtained.

 In a curved background, the inhomogeneous field equations
 of NLED
 read
\be \label{eq:fieldNLEDcurved}
\p_\mu\paren{\sqrt{-g}E^{\mu\nu}}=-4\pi\sqrt{-g}j^\nu\ ,
\en
where $E^{\mu\nu}$
is defined in terms of $\mathbf{D}$ and $\mathbf{H}$
.

Let us now introduce 
the isotropic coordinates $\paren{t,\riso,\theta,\varphi}$ for Schwarzschild's spacetime, defined by 
$$r=\Sigma\paren{\riso}\riso,\ \text{with}\ \Sigma\paren{\riso}=\paren{1+\frac{r_s}{4\riso}}^2\ ,$$
where $r$ is the radial coordinate in standard Schwarzschild coordinates and $r_s=2M$ is the Schwarzschild radius. In these coordinates,
for an arbitrary theory of NLED
with Lagrangian density  
$\mathcal{L}=\mathcal{L}\paren{F,G}$
, {in the electrostatic case},  the excitation tensor {defined in Eq. \eqref{eq:defexcitation}} takes the form $E^{\mu\nu}=\p\mathcal{L}/\p F_{\mu\nu}=2\mathcal{L}_F F^{\mu\nu}$, {since 
$G=0$ in electrostatics}. Consequently, taking into account that the homogeneous equation guarantees that $\mathbf{E}=-\boldsymbol{\nabla}\phi$,
we can write 
Eq. \eqref{eq:fieldNLEDcurved} in terms of the electrostatic potential {$\phi$} as
follows:
\be\label{eq:potNLEDcurv}
\Delta\phi+\frac{\paren{1-r_s/4\riso}}{\paren{1+r_s/4\riso}^3}\frac{\p}{\p\riso}\paren{\frac{\paren{1+r_s/4\riso}^3}{\paren{1-r_s/4\riso}}}\frac{\p\phi}{\p\riso}
=\frac{2\pi\rho}{\mathcal{L}_{F}\paren{\boldsymbol{\nabla}\phi}}\paren{1-\frac{r_s^2}{16\riso^2}}^2-\frac{1}{\mathcal{L}_F\paren{\boldsymbol{\nabla}\phi}}\boldsymbol{\nabla}\phi\cdot\boldsymbol{\nabla}\mathcal{L}_F\ .
\en
This is a nonlinear equation for the electrostatic potential, 
as a consequence of the nonlinear (in general) dependence of the excitation tensor with $F^{\mu\nu}$. When the theory under consideration is Maxwell's electrodynamics, \textit{i.e.,} $\mathcal{L}_F=-1/2$, the resulting electrostatic equation reads
\be
\label{phim}
\Delta\phi_M+\frac{\paren{1-r_s/4\riso}}{\paren{1+r_s/4\riso}^3}\frac{\p}{\p\riso}\paren{\frac{\paren{1+r_s/4\riso}^3}{\paren{1-r_s/4\riso}}}\frac{\p\phi_M}{\p\riso}=-4\pi\rho\paren{1-\frac{r_s^2}{16\riso^2}}^2
\en
 The solution of this equation is the potential of a static point charge in Maxwell's EM in Schwarzschild's spacetime. The analytic form of this  solution 
found by Copson \cite{Copson1928} was reformulated by Linet \cite{Linet1976}, who added a monopole term inside the horizon, in order to account for the correct asymptotic limit. Linet's  solution for a particle with charge $e$ located at a position $\riso=\aiso$ is
\be\label{eq:psiLin}
\phi_{M(L)}=\frac{e\Sigma\paren{\aiso}^{-1}}{\riso\Sigma\paren{\riso}\mu}\paren{\mu+\frac{r_s}{4\aiso}}^2\ ,
\en
with 
$$\mu\equiv\sqrt{\frac{\paren{\riso-\biso}^2+2\biso\riso\paren{1-\cos\theta}}{\paren{\riso-\aiso}^2+2\aiso\riso\paren{1-\cos\theta}}}\ , \ \biso\equiv\frac{r_s^2}{16\aiso}\ .$$
The potential \eqref{eq:psiLin} can be decomposed as a linear superposition of two solutions. Defining
\begin{align}
\phi_{M(C)}&= \frac{e}{\aiso\Sigma(\aiso)\riso\Sigma(\riso)}\paren{\mu\aiso+\frac{\biso}{\mu}}\ , \label{eq:Copson}\\
\phi_{M(p)}&=\frac{e r_s}{2\aiso\Sigma(\aiso)\riso\Sigma(\riso)} \label{eq:Vp}\ ,
\end{align}
it follows that
\be
\phi_{M(L)}=\phi_{M(C)}+\phi_{M(p)}\ ,
\en
where $\phi_{M(C)}$ is the solution originally obtained by Copson \cite{Copson1928}. 
The second term can also be interpreted as arising from the fact that the black hole horizon acts as a conducting surface, therefore inducing an image particle inside the horizon that its ultimately responsible for the potential $\phi_{M(p)}$. This phenomenon is also known as black hole polarization and has interesting consequences \cite{Bekenstein1999,Hod1999}.

Let us now turn to the case of a static charged source for a  generic theory of NLED in Schwarzschild's background. The electric potential obeys Eq. \eqref{eq:potNLEDcurv}.
The solutions of which  should, in principle, be obtained by resorting to 
numerical integration and/or analytical approximations. However, for the case of an arbitrary theory of NLED in Schwarzschild spacetime, it was shown in \cite{Falciano2019} that the following theorem holds:
\begin{theorem}\label{thm:1}
The electrostatic potential $\phi(x)$ produced by a charged particle satisfying a generic NLED theory $\mathcal{L}(F,G)$ in a static, spherically symmetric spacetime is entirely specified by the electrostatic potential $\phi_M(x)$ satisfying the equations of Maxwell's electrostatics in the same background, such that $\phi_M(x)$ and $\phi(x)$ have the same asymptotic behaviour. The displacement vector, given by $\mathbf{D}=-\boldsymbol{\nabla}\phi_M(x)$, is curl-free.
\end{theorem}

\begin{proof}
   In the electrostatic case, $\mathbf{B}=0$ and 
   there is no time dependence. The field equation (\ref{eq:fieldNLEDcurved}) reads
    \be \label{eq:thm1}
    \p_\mu\paren{\sqrt{-g}D^\mu}=-4\pi\sqrt{-g}\rho\ ,
    \en
    where the density is defined as $\rho=j_\alpha v^\alpha$, and we have used the fact that the four-velocity of a static particle, $v^\mu=c\delta^{\mu}_{\ 0}/\sqrt{g_{00}}$ satisfies $\nabla_\mu v^\mu=\p_\mu v^\mu=0$.
    
    Using $D^\mu=-2\mathcal{L}_F\paren{E}E^\mu$ 
    (see Eq. \eqref{eq:defexcitation})
    in Eq. \eqref{eq:thm1},
with the electric field given by $E^\mu=\paren{0,-\boldsymbol{\nabla}\phi}$ in the reference frame where the particle is at rest, it follows that the electrostatic potential obeys Eq. \eqref{eq:potNLEDcurv}. 
    
Let $\psi(x)$ be an auxiliary scalar function defined as 
\be\label{eq:thm2}
\psi(x)=-2\int\mathcal{L}_F\paren{\boldsymbol{\nabla}\phi}\boldsymbol{\nabla}\phi\cdot\dd\mathbf{l}\ ,
\en
where the integral is to be calculated along a path with tangent vector $\dd\mathbf{l}$. A straightforward calculation shows that if $\psi$ satisfies Eq. \eqref{phim}, then $\phi$ is a solution of Eq. \eqref{eq:potNLEDcurv}. Hence, $\psi\equiv \phi_M$, and the displacement vector field reads $\mathbf{D}=
-\boldsymbol{\nabla}\phi_M=
2\mathcal{L}_F\paren{\boldsymbol{\nabla}\phi}\boldsymbol{\nabla}\phi$.  Furthermore, assuming that the constitutive equations are invertible, the $P$ framework allows us to write the Faraday tensor as a function of the excitation tensor and its dual. Still in the electrostatic case, it follows from Eq. \eqref{eq:Hframe} that $\boldsymbol{\nabla}\phi=-2\mathcal{H}_P\paren{\boldsymbol{\nabla}\psi}\boldsymbol{\nabla}\psi$. 
\end{proof}

The theorem allows us to express the electric displacement field, $\mathbf{D}$
as the gradient of the  potential for the same problem in Maxwell's theory, namely $\mathbf{D}=-\boldsymbol{\nabla}\phi_M$, where $\phi_M$ is given by Eq. \eqref{eq:psiLin}. For a given theory of NLED, the electric field can be calculated using Eq. \eqref{eq:FrelE}. In particular, 
the exact expression for the electric field in BI theory can be obtained from
Eq.\eqref{eq:EBvectors}, yielding
\be\label{eq:EexactBI}
\mathbf{E}=-\boldsymbol{\nabla}\phi=-\frac{\boldsymbol{\nabla}\phi_M}{\sqrt{1+\vert\boldsymbol{\nabla}\phi_M\vert^2\beta^{-2}}}\ .
\en

\subsection{Black hole polarization and approximated expression for the potential}

The electrostatic potential $\phi$ 
for a BI particle in Schwarzschild's background, follows by numerical integration of Eq. \eqref{eq:EexactBI}. Let us show now that an approximated expression for $\phi$ 
near and at the particle can be obtained as discussed in \cite{Falciano2019}.  First, note that the exact expression for the electric field, Eq. \eqref{eq:EexactBI}, has a weak-field limit that coincides with that of  Maxwell's theory. Hence, far away from the particle the well-known Maxwell behavior is recovered. On the other hand, since the potential $\phi_M$ vanishes on the horizon, the electric field (given by Eq. \eqref{eq:EexactBI}) will also vanish there, independently of the position of the charge. This is in agreement with the interpretation of the horizon as a conducting surface, and hence the idea of black hole polarization is still relevant in this context.

Since, in BI electrodynamics, nonlinearities become {important} near the sources, one may wonder what happens with the nonlinearities of both the particle and its image charge when the particle is located near the horizon. Black hole polarization ensures us 
that the nonlinearity regions will remain confined both outside and inside the horizon for the real and image particle correspondingly. This is illustrated in Figure \ref{fig:BIreg}.
\begin{figure}[ht]
\centering
\includegraphics[width=0.47\textwidth]{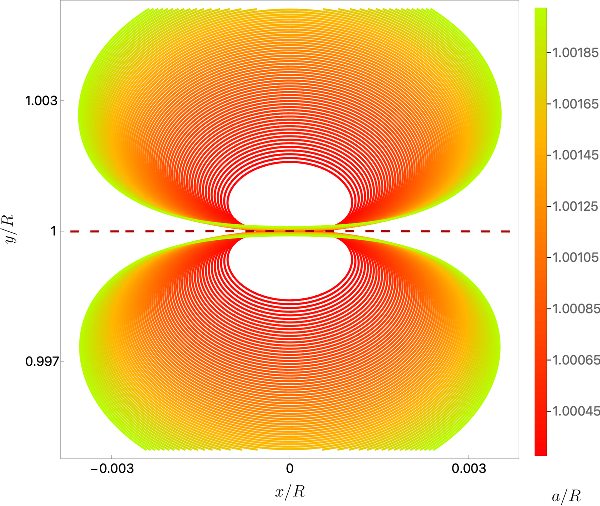}
\caption{The nonlinearity region,  modelled, in Cartesian coordinates, as the region where {$\vert\mathbf{D}\vert=\vert\boldsymbol{\nabla}\phi_{M(C)}\vert\sim\beta$}, where {$\phi_{M(C)}$} is given by Eq. \eqref{eq:Copson}, for $e=1$, $\beta=3$. All quantities have been normalized by the black hole's mass. Each coloured contour displays the nonlinearity region for a particular position of the charge, $a/R$. It is seen that as the charge approaches the horizon the region shrinks and never touches the horizon, for both for the real particle and its image charge inside the horizon.
}
\label{fig:BIreg}
\end{figure}

The fact that black hole polarization is still present
when considering BI 
electrodynamics implies that the polarization contribution $\phi_{M(p)}$ in Eq. \eqref{eq:Vp} will still be present. Therefore we can approximate the electric field outside the horizon as a sum of two terms as
\be \label{eq:apEBIapprox}
\mathbf{E}\paren{\riso,\theta}=-\boldsymbol{\nabla}\phi\approx-\frac{\boldsymbol{\nabla} \phi_{M(C)}}{\sqrt{1+\vert\boldsymbol{\nabla} \phi_{M(C)}\vert^2\beta^{-2}}}-\boldsymbol{\nabla} \phi_{M(p)}\ .
\en
The first term appearing on the r.h.s. of Eq. \eqref{eq:apEBIapprox} is the regularized
expression for the electric field coming from Copson's solution while the second term corresponds to the black hole polarization contribution to the electric field. Eq. \eqref{eq:apEBIapprox} leads to the following decomposition of the potential:
\be \label{eq:appotBIapprox}
\phi\paren{\riso,\theta}\approx\bar{\phi}\paren{\riso,\theta}+\phi_{M(p)}\paren{\riso,\theta}\ ,
\en
where $\bar{\phi}(\riso,\theta)$ is given by
\be
\bar{\phi}=\int\frac{\boldsymbol{\nabla}\phi_{M(C)}}{\sqrt{1+\vert\boldsymbol{\nabla}\phi_{M(C)}\vert^2\beta^{-2}}}\cdot\dd\mathbf{l}\ .
\en
 Unfortunately, the functional form of $\phi_{M(C)}$ (see Eq. \eqref{eq:Copson}) forbids analytical integration,  but
we can resort to a further approximation. Copson's potential \eqref{eq:Copson} diverges as the position of the charge is approached, namely as $\riso\to\aiso$. Therefore we can analyze the behavior of Copson's potential as it starts diverging (i.e., as $\mu\to\infty$ in Eq. \eqref{eq:Copson}), thus obtaining\footnote{See the appendix in \cite{Falciano2019} for details on the approximations  performed here.}
\be\label{eq:apapproxCopsonfield}
\lim_{\riso\to\aiso}\boldsymbol{\nabla}\phi_{M(C)}\approx\frac{\phi_{M(C)}}{\mu}\boldsymbol{\nabla}\mu\quad \Rightarrow\qquad \vert\boldsymbol{\nabla}\phi_{M(C)}\vert\approx\beta\frac{\phi_{M(C)}^2}{\zeta_{\aiso}^2}\ ,
\en
where $\zeta_{\aiso}\equiv\sqrt{e \beta}\paren{g_{00}\paren{\aiso}}^{1/4}$. With the approximation \eqref{eq:apapproxCopsonfield} it is possible to integrate Eq. \eqref{eq:apEBIapprox} and obtain, after imposing that the asymptotic behavior of the potential coincides with that of Maxwell's, that the electrostatic potential for a BI particle immersed in a Schwarzschild background {can be approximated in the neighbourhood of the particle by}
\be
\bar{\phi}\paren{\riso,\theta}=\zeta_{\aiso} \paren{\frac{\Gamma\paren{\frac{1}{4}}^2}{4\sqrt{\pi}}-\frac{\zeta_{\aiso}}{\phi_{M(C)}}\ _2F_1\left[\frac{1}{2},\frac{1}{4};\frac{5}{4};-\paren{\frac{\zeta_{\aiso}}{\phi_{M(C})}}^4\right]}\ .
\en

\section{BI particle near the horizon of a Schwarzschild black hole as a BI particle in Rindler spacetime} \label{sec:Rindler1}

We are interested in the field of an accelerated Born-Infeld source in flat spacetime. As previously stated, the usual approach to obtain a solution for the potential of the accelerated source is not feasible in this case via the Green's function method. Nonetheless, we can exploit the equivalence between a static observer close to the horizon of a Schwarzschild black hole and an accelerated observer in flat spacetime: both are characterized by Rindler {coordinates}. 

The starting point is the expansion of  the Schwarzschild metric near the horizon as (see \ref{app:PotRind} for a detailed derivation)
\be \label{eq:dslambda}
\dd s^2=2\kappa_s\lambda\dd t^2-\frac{1}{2\kappa_s\lambda}\dd\lambda^2-\dd\rho^2-\rho^2\dd\varphi^2\ ,
\en
where $\kappa_s$ is the 
surface gravity of Schwarzschild's black hole . This metric reduces to Minkowski spacetime when $\lambda=\lambda_0=1/2\kappa_s=r_s$. Then, in this coordinate system the particle is at rest at $\lambda=\lambda_0$. By performing the corresponding expansions, Copson's potential, Eq. \eqref{eq:Copson}, for the near-horizon approximation reads
\be \label{eq:Copsonlambda}
\phi_{M(C)}\paren{\lambda}=\frac{e}{2r_s}\frac{\lambda+\lambda_0+\rho^2\kappa_s/2}{\sqrt{\left[\lambda-\lambda_0+\rho^2\kappa_s/2\right]^2+2\kappa_s\lambda_0\rho^2}}\ .
\en
It is important to point out that this expression is the exact electric potential of a uniformly accelerated charged particle in Maxwell's theory in a comoving frame \cite{Gupta1998}.

The solution for the electrostatic potential in a region arbitrarily close to the horizon of a Schwarzschild black hole is given by Eq. \eqref{eq:appotBIapprox}, where $\bar{\phi}$ is to be identified with $\phi_{M(C)}\paren{\lambda}$ in Eq. \eqref{eq:Copsonlambda}. Now, it is important to stress that when approximating Schwarzschild to Rindler spacetime, a certain amount of information is lost due to neglecting the spacetime curvature effects on the particle's field acting over itself \cite{Smith1980,Poisson2011}, which can alternatively be interpreted as losing the contribution on the potential coming from the effect of black hole polarization, \textit{i.e.}, $\phi_{M(p)}$ will no longer be present when considering Rindler spacetime \cite{MacDonald1985}. Hence, since we are solely interested in the Rindler's limit for Schwarzschild spacetime, the electric field is given exactly by the relation
\be \label{eq:EfieldRin}
\mathbf{E}=-\boldsymbol{\nabla}\phi=-\frac{\boldsymbol{\nabla} \phi_{M(C)}}{\sqrt{1+\vert\boldsymbol{\nabla} \phi_{M(C)}\vert^2\beta^{-2}}}\ , 
\en
and, in order to obtain the electrostatic potential, it suffices to integrate 
\be \label{eq:elecpotRin}
\phi=-\int^{r}_{\infty}\mathbf{E}\cdot\dd\mathbf{l}\ .
\en
Let us use the expressions obtained above to calculate some relevant quantities. Figure \ref{fig:nomElambda} shows the norm of the electric field, given by Eq. \eqref{eq:EfieldRin},
at different positions with respect to the charge (located at $\lambda=\lambda_0)$. As expected, the electric field is regular at
the position of the charge and it coincides with the usual Maxwell solution far from the charge. 

In order to obtain an exact expression for 
the electrostatic potential in the system comoving with the accelerated particle,  we can follow the approach used in the previous section for Schwarzschild spacetime and directly integrate the electric field.
The result is (see \ref{app:BISbh} for details)
\be \label{eq:apppot}
\phi(\lambda)=\Delta\paren{\frac{\Gamma\paren{\frac{1}{4}}^2}{4\sqrt{\pi}}-\frac{\Delta}{\phi_{M(C)}(\lambda)}\ _2F_1\left[\frac{1}{2},\frac{1}{4},\frac{5}{4};-\paren{\frac{\Delta}{\phi_{M(C)}(\lambda)}}^4\right]}\ ,
\en
where $\Delta\equiv\sqrt{e\beta}$.
At first glance, it may seem that the spacetime geometry plays no role when calculating the potential at the position of the particle. Nevertheless, 
the redshift must be taken into account (as in the Schwarzschild case discussed above). For the coordinate system used it is simply $g_{00}(\lambda_0)=1$. Moreover, the value of the potential at the {location of the} charge is
\be
\phi\paren{\lambda_0}=\sqrt{e\beta}\frac{\Gamma\paren{\frac{1}{4}}^2}{4\sqrt{\pi}}\ ,
\en
which coincides with the value of the electrostatic potential of a static BI particle in flat spacetime (see Eq. \eqref{eq:potBIfst}). 

Figure \ref{fig:relerrorpot} displays the relative error of the approximation given by Eq. \eqref{eq:apppot} compared to the value obtained by numerical integration of \eqref{eq:EfieldRin}, for the potential at the 
position of the particle. It is seen that the error is almost independent of the value of the electric charge
for small values of the charge, and that for high values of the maximum field parameter $\beta$, it drastically decreases. In particular, for large values of $e$ and $\beta$, the relative error is approximately $10^{-4}$.

To calculate the nontrivial components of the Faraday tensor, the following derivatives are needed:
\begin{align}
\p_\lambda\phi&=\frac{1}{\sqrt{1+\Phi^4}}\p_\lambda \phi_{M(C)}\ , \label{eq:dlambdaphi}\\
\p_\rho\phi&=\frac{1}{\sqrt{1+\Phi^4}}\p_\rho \phi_{M(C)}\ , \label{eq:drhophi}
\end{align}
where 
$\Phi\equiv \phi_{M(C)}/\Delta$, and certain properties for the derivatives of the hypergeometric function were used \cite{Lebedev1965}. It follows that the non-trivial components of the Faraday tensor are given in exact form by:
\begin{align}
F^{\lambda0}=-F_{\lambda0}=-\p_\lambda\phi\ , \label{eq:42}\\
F^{\rho0}=-\frac{1}{2\kappa_s\lambda}F_{\rho0}=-\frac{1}{2\kappa_s\lambda}\p_\rho\phi\ ,\label{eq:43}
\end{align}
and the electric field components, given by Eq. \eqref{eq:Lorentz}, are
\begin{align}
E^{\lambda}&=F^{\lambda}_{\ 0}/\sqrt{g_{00}} \ ,\\
E^{\rho}&=F^{\rho}_{\ 0}/\sqrt{g_{00}}\ .
\end{align}
We analyze next the behavior of the electric field near the charged particle. Since the potential is regular at the 
position of the particle and, as stated before, the redshift factor there is simply $g_{00}(\lambda_0)=1$, we arrive at the following expressions:
\begin{align}
E^\lambda&=\frac{2\kappa_se\lambda_0}{\omega^3}\sqrt{2\kappa_s\lambda}\frac{\lambda-\lambda_0-\rho^2\kappa_s/2}{\sqrt{1+\Phi^4}}\ , \label{eq:Elambda}\\
E^\rho&=\frac{2\kappa_se}{\omega^3}\lambda_0\rho\sqrt{2\kappa_s\lambda}\frac{1}{\sqrt{1+\Phi^4}}\ , \label{eq:Erho}
\end{align}
where {we have defined}
\be
\omega\equiv\sqrt{\left[\lambda-\lambda_0+\rho^2\kappa_s/2\right]^2+2\kappa_s\lambda_0\rho^2}\ .
\en
Let us note that, contrary to  Maxwell's case, when 
the position of the particle is approached, namely {as} $\rho\to0$ and $\lambda\to\lambda_0$, only regular terms are present, given by 
\begin{align}
E^\lambda&\approx\beta+\mathcal{O}\paren{\rho^2}\ , \\
E^\rho&\approx\mathcal{O}\paren{\rho}\ .
\end{align}

\begin{figure}[ht]
\centering
\includegraphics[width=0.47\textwidth]{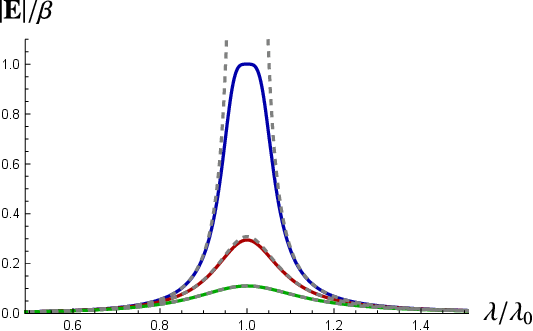}
\caption{Plot of the norm of the electric field for Born-Infeld (solid lines) and Maxwell electrodynamics 
(dashed lines) in the $\paren{\lambda,\rho}$ coordinate system. The BI electric field is given by \eqref{eq:EfieldRin} 
while the Maxwell counterpart is just the gradient of 
Copson's potential . We have set $e=0.01$ and $\rho=0.001$ (blue curve), $\rho=0.18$ (red curve) and $\rho=0.3$ (green curve).
}
\label{fig:nomElambda}
\end{figure}

\begin{figure}[ht]
\centering
\includegraphics[width=0.47\textwidth]{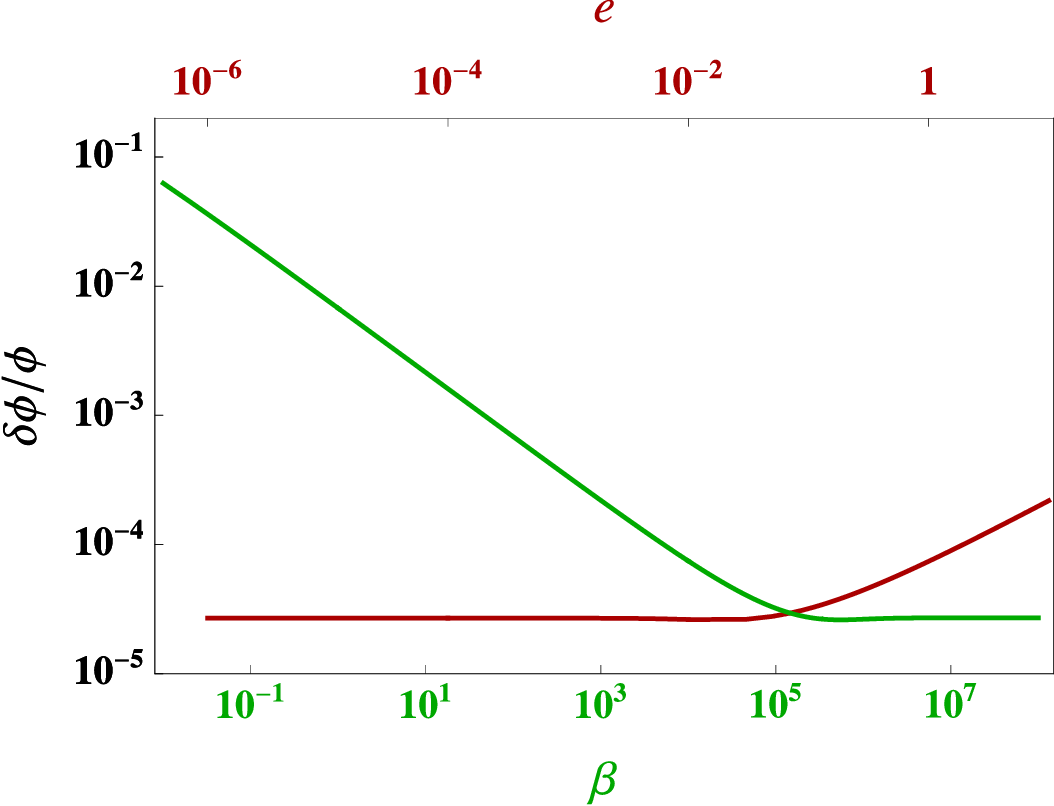}
\caption{Relative error of the approximation of the potential at the position of the particle, given by Eq. \eqref{eq:apppot}, with respect to the exact potential obtained by numerically integrating Eq. \eqref{eq:EfieldRin}. 
The relative error was calculated as a function of $\beta$ for $e=0.01$,
(green plot) and $e$ 
for $\beta=10^6$
(red plot). The error drastically decreases for large values of $\beta$, while small values of $e$ do not contribute significantly to the error. Geometrized units were used in the plot, and  we have set $\rho=10^{-4}$.
}
\label{fig:relerrorpot}
\end{figure}

\section{Transformation to the inertial coordinate system}\label{sec:inertcomp}

The calculations 
presented in the previous section
are valid for a BI particle arbitrarily close to the horizon of a Schwarzschild black hole. At this limit, the static particle outside the horizon behaves exactly like a linearly accelerated particle in flat spacetime and is described in terms of Rindler coordinates adapted to the particle. 

Gupta and Padmanabhan \cite{Gupta1998} considered an accelerated charged particle in Maxwell's theory in comoving (Rindler) coordinates, and showed that the corresponding field components expressed in inertial coordinates are those of the Liénard-Wiechert field of a linearly accelerated charge. In this section we shall follow Ref. \cite{Gupta1998} in order to obtain the relevant electromagnetic field components in the laboratory frame for a linearly accelerated BI source, expressed in Rindler (comoving) coordinates in Eqs. \eqref{eq:Elambda}-\eqref{eq:Erho}.

{\subsection{Coordinate transformation}}
\label{cotra}
In brief, let us consider the motion of the accelerated particle in two different systems. In what follows, the system $S$ is an inertial system 
with respect to which the particle is accelerated; while the system $S'$ is a non-inertial system adapted to the accelerated particle.

Let $\{x^{\mu}\}=\paren{t,z,\rho,\phi}$ be the coordinates corresponding to the inertial system $S$. We shall consider an event $P$ (at which the value of the fields is to be calculated) with coordinates $\paren{t,z,\rho,\phi}$, and the retarded event $O$, with coordinates $\paren{t_0,z_0,0,0}$. At $O$, the proper acceleration is given by
\be
g=\sqrt{-a^\mu a_\mu}=a_{\ret}\gamma^3\ ,
\en
where $\gamma\equiv\paren{1-v_{\ret}^2}^{-1/2}$, and $v_{\ret}$, $a_{\ret}$ are the velocity and acceleration at event $O$, respectively. Let us introduce another system, $S'$, with coordinates 
$\{\xi^\alpha\}=\paren{\tau,\lambda,\rho,\phi}$, which coincides with the world line of the charge up to $v^\mu$ and $a^\mu$ at $O$, in such a way that at $O$, in the frame $S'$, the charge is instantaneously at rest without acceleration, and the constant, proper, acceleration of $S'$ is $g$.
In particular, let us assume that the uniformly accelerated motion takes place along the $z$ axis. The transformation between $S$ and $S'$ is
\begin{align}
t'&=\frac{\sqrt{2g\lambda}}{g}\sinh\paren{g\tau}\ , \label{eq:tprime}\\
z'&= \frac{\sqrt{2g\lambda}}{g}\cosh\paren{g\tau}\ ,\label{eq:zprime}
\end{align}
where we have defined the new inertial coordinates $t'$ and $z'$ as
\begin{eqnarray}
    t'&=&t-t_0+\frac{\gamma v_{\ret}}{g}\ , \\
    z'&=&z-z_0+\frac{\gamma}{g}\ .
\end{eqnarray}
In these coordinates, the event $O$ takes place at
\begin{align}
&t'_0=\frac{\gamma v_\ret}{g}\ , \ z'_0=\frac{\gamma}{g}\ ,\label{eq:t0}\\
\Rightarrow &\lambda_0=\frac{1}{2g}\ , \ \tau_0=\frac{1}{g}\sinh^{-1}\paren{\gamma v_\ret}\ . \label{eq:eventOlambda}
\end{align}
Let us define the distance between a point $P$ and the event $O$ in $S$ by the null vector
\be\label{eq:Rmu}
R^\mu=\paren{t'-t'_0,z'-z'_0,\rho,\phi}\ .
\en
From $R^\mu R_\mu=0$ it follows that
\be \label{eq:coshid}
\cosh g\paren{\tau-\tau_0}=\frac{1}{2\sqrt{\lambda\lambda_0}}\paren{\lambda+\lambda_0+g\rho^2/2}\ ,
\en
and, since the non-trivial components of the velocity are
\begin{align}
v^0_\ret&=\sqrt{2g\lambda_0}\cosh\paren{g\tau_0}\ , \label{eq:v0ret}\\
v^z_\ret&=\sqrt{2g\lambda_0}\sinh\paren{g\tau_0}\ ,\label{eq:vzret}
\end{align}
we obtain
\be\label{eq:74}
R_\mu v^\mu_\ret=2\sqrt{\lambda\lambda_0}\sinh g\paren{\tau-\tau_0}\ .
\en
Using the identity
given by Eq. \eqref{eq:coshid}, in the latter equation,
\begin{align}
    R_\mu v^\mu_\ret&=\sqrt{\paren{\lambda-\lambda_0+g\rho^2/2}^2+2g\rho^2\lambda_0} \label{eq:75} \\
    &=\frac{g}{2}\sqrt{\paren{z'^2-t'^2-\paren{1/g}^2+\rho^2}^2+4\rho^2/g^2}\ . \label{eq:66}
\end{align}
Additionally, applying a Lorentz transformation into Eqs. \eqref{eq:v0ret} and \eqref{eq:vzret} we obtain that the time component of the four-velocity in an arbitrary local inertial frame is given by 
\be
v^0_\mathrm{iner}=\sqrt{2g\lambda}\cosh g\paren{\tau-\tau_0}\ .
\en
Let us apply these transformations to 
the potential. 
Since we are interested in mapping Copson's solution  to an inertial reference frame and adapting it to the case of a linearly accelerated particle, we can replace the black hole's surface gravity with a parameter $g$ that denotes the particle's proper acceleration, \textit{i.e.}, $\kappa_s\to g$
in Eq. \eqref{eq:Copsonlambda}, yielding
\be \label{eq:Copsonlambda1}
\phi_{M(C)}(\lambda)=eg\frac{\lambda+\lambda_0+\rho^2g/2}{\sqrt{\left[\lambda-\lambda_0+\rho^2g/2\right]^2+2g\lambda_0\rho^2}}\ .
\en
From Eqs. \eqref{eq:tprime} and \eqref{eq:zprime}, $\phi_{M(C)}$ in the coordinates of the inertial system takes the form
\be \label{eq:Copsoninert}
\phi_{(M)C}\paren{t',z'}=eg\frac{z'^2-t'^2-\paren{1/g}^2+\rho^2}{\sqrt{\paren{z'^2-t'^2-\paren{1/g}^2+\rho^2}^2+4\rho^2/g^2}}\ .
\en
Using the relations presented above, we can write this potential,  in an arbitrary inertial frame as
\be
\label{lwpot}
\phi_{M(C)}=e\frac{v^0_\mathrm{iner}}{R_\mu v^\mu_\ret}\ ,
\en
which is precisely the known result for the scalar potential in LW solution [see \ref{app:LWpot}].

\subsection{Transformation of the fields}
Given the transformation between the inertial, $\{x^\mu\}$, and comoving coordinates, $\{\xi^\mu\}$,
 the Faraday and excitation tensors in {both coordinate systems} are related  by the transformation
\be\label{eq:LWapproximpl}
\mathcal{F}^{\mu'}_{\ \nu'}(x)=\frac{\p x^{\mu'}}{\p \xi^{\alpha}}\frac{\p\xi^{\beta}}{\p x^{\nu'}}F^{\alpha}_{\ \beta}(\xi)\ ,
\en
\be\label{eq:transEmunu}
\mathcal{E}^{\mu'}_{\ \nu'}(x)=\frac{\p x^{\mu'}}{\p \xi^{\alpha}}\frac{\p\xi^\beta}{\p x^{\nu'}}E^{\alpha}_{\ \beta}(\xi).
\en
Furthermore, in view of the constitutive relation (Eq. \eqref{eq:FrelE}), Eq. \eqref{eq:LWapproximpl} can be expressed in terms of the excitation tensor as
\be \label{eq:Fmunuinert}
\mathcal{F}^{\mu'}_{\ \nu'}(x)=-\frac{\p x^{\mu'}}{\p \xi^{\alpha}}\frac{\p\xi^{\beta}}{\p x^{\nu'}}\frac{1}{\sqrt{V}}\paren{E^{\alpha}_{\ \beta}+\frac{S}{2\beta^2}\widetilde{E}^{\alpha}_{\ \beta}}\ .
\en
The field $\mathcal{F}^{\mu'}_{\ \nu'}$
is the analogous in BI NLED, for constant acceleration, of the Liénard-Wiechert field in Maxwell's theory. In order to obtain such field, it suffices to know the components of the excitation tensor, $E^{\mu\nu}$, in the comoving frame. In the case at hand, 
this is equivalent to the calculation of $\mathbf{D}$, 
which follows from the results of Theorem \ref{thm:1} quoted in Sec. \ref{ssec:BIpart}  (and originally presented in \cite{Falciano2019}). The theorem guarantees that once the  potential $\phi_M$, is known in the comoving frame, the calculation of the components of $E^{\mu\nu}$ is straightforward. In a general 
case, the calculation of $\mathbf H$ would also be needed. We shall take instead another path to obtain the components of the Faraday tensor in the inertial frame. In the last section we have explicitly calculated the nontrivial components of $F^\alpha_{\ \beta}(\xi)$ in the comoving frame. We now proceed to use these results to calculate the field components in the laboratory frame using the results for the derivatives of the approximated potential (\eqref{eq:dlambdaphi} and \eqref{eq:drhophi}) in Eqs. \eqref{eq:42} and \eqref{eq:43}, and transform the latter using \eqref{eq:LWapproximpl}. The non-trivial components of the Faraday tensor in the comoving frame, $F^\alpha_{\ \beta}\paren{\xi}$, are those associated with the $\lambda$ and $\rho$ components of the field. Therefore the relevant components of the electromagnetic field in the laboratory frame are, explicitly
\begin{align}
\mathcal{F}^{z'}_{\ t'}&=\frac{4e}{g^2}\frac{1}{\sqrt{1+\Phi^4}}\frac{\left(z'^2-t'^2-(1/g)^2-\rho^2\right)}{\left(\left(z'^2-t'^2-(1/g)^2+\rho^2\right)^2+4\rho^2/g^2\right)^{3/2}}\ , \label{eq:Finertzt}\\
\mathcal{F}^{\rho}_{\ t'}&=\frac{1}{\sqrt{1+\Phi^4}}\frac{8e\rho z'}{g^2\left((z'^2-t'^2-(1/g)^2+\rho^2)^2+4\rho^2/g^2\right)^{3/2}}\ ,\label{eq:Finertrhot}\\
\mathcal{F}^{z'}_{\ \rho}&=\frac{1}{\sqrt{1+\Phi^4}}\frac{8e\rho t'}{g^2\left((z'^2-t'^2-(1/g)^2+\rho^2)^2+4\rho^2/g^2\right)^{3/2}}\ .\label{eq:Finertzrho}
\end{align}
Note that the transformation given in Eq. \eqref{eq:Fmunuinert}, which yields the Faraday tensor components in the inertial coordinate system (Eqs. \eqref{eq:Finertzt}-\eqref{eq:Finertzrho}), is ultimately proportional to the derivatives of the electrostatic potential in the comoving (non-inertial) frame. Hence, in order to test our approximation for the electrostatic potential, Eq. \eqref{eq:apppot}, we compute the relative error of the norm of the electric field obtained from the exact expression, given by Eq. \eqref{eq:EfieldRin}
expressed in the inertial coordinates, 
with respect to the norm of the electric field obtained from the approximated field components (Eqs. \eqref{eq:Finertzt} and \eqref{eq:Finertrhot},
where the potential is given by 
Eq. \eqref{eq:Copsoninert}). The result {of this calculation} is displayed in 
Figure \ref{fig:relerrorEinert}, which shows that the approximation is in good agreement with the exact result.

Let us next check how the approximation for the fields works in the homogeneous equation (Eq. \eqref{eq:Bianchi}). Upon substitution of Eqs. \eqref{eq:Finertzt} - \eqref{eq:Finertzrho}, it follows that Eq. \eqref{eq:Bianchi} is satisfied to $\mathcal{O}(\rho)$ if we consider observation points near the worldline of the particle and to $\mathcal{O}(\rho^{-7})$ for points very far from the particle, where $\rho$ is the distance between the observation point and the worldline of the particle. 
Hence, in both 
the applications to be developed below, the approximation is valid. Namely, in radiation reaction (where
the fields are evaluated on the particle) and in the radiation field (where the fields are
evaluated very far from the particle), the homogeneous equation is satisfied by the approximation. 
\begin{figure}[ht]
\centering
\includegraphics[width=0.47\textwidth]{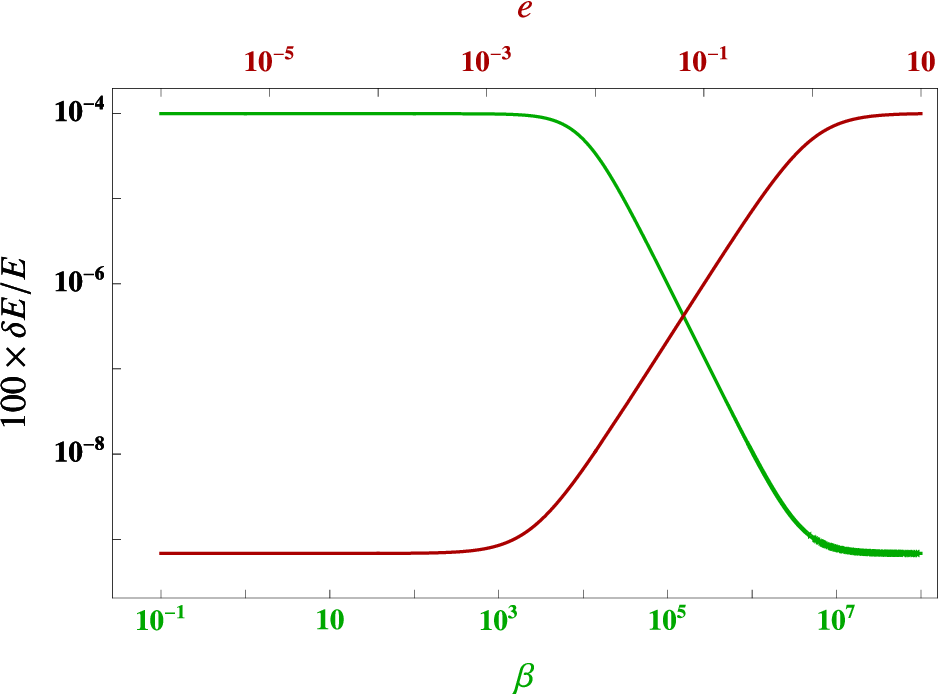}
\caption{Percent error of the approximation of the norm of the electric field obtained 
using Eq. \eqref{eq:apppot} for the potential
with respect to the exact result calculated using \eqref{eq:EfieldRin}, at the particle's position in the inertial coordinate system. The red curve represents the error for running $e$ and $\beta=10^{6}$, while the green curve represents the running $\beta$ for $e=0.01$. In both plots we have taken $g=1$, $\rho=0.001$ and $t'=0.1$.}
\label{fig:relerrorEinert}
\end{figure}

\section{Radiation and Radiation Reaction}\label{sec:Radiation}

Having at our disposal the components of the tensor $\mathcal{{F}}^{\mu'}_{\ \nu'}$, we shall analyze next the radiation field and how it acts on the charged particle. 

\subsection{Radiation}

In order to study the features of the radiation field generated by a linearly accelerated particle in BI NLED, 
the Poynting vector will be calculated next, using the field components in an inertial frame, given by Eq. \eqref{eq:Fmunuinert}.
From Eq. \eqref{eq:Tmn}, it follows that the Poynting vector for the nonlinear theory under consideration is 
\be
\mathbf{S}=\frac{1}{4\pi}\mathbf{E}\times\mathbf{H}=\frac{1}{4\pi}\frac{\mathbf{E}\times\mathbf{B}}{\sqrt{1+F/\beta^2}},
\en
where the vectors $\mathbf{E}$ and $\mathbf{B}$ in the inertial frame are displayed in Eqs. \eqref{eq:Finertzt}-\eqref{eq:Finertzrho}
We can further rewrite the electric and magnetic field components using the notation 
introduced in Sect.\ref{cotra}. 
By plugging Eq. \eqref{eq:eventOlambda} into Eqs. \eqref{eq:v0ret} and \eqref{eq:vzret}, the expression for $R_\mu v^\mu_\ret$ in Eq. \eqref{eq:66} can be rewritten as follows:
\be\label{eq:newRmuvmu}
R_\mu v^\mu_\ret=\gamma\paren{R-v_\ret R_{z'}}\ ,
\en
where $R=R_0=t'-t'_0$ 
and $R_{z'}=z'-z'_0$. 
Taking into account
Eqs. \eqref{eq:Rmu} and \eqref{eq:t0}, and the condition $R^\mu R_\mu=0$ the following identities
are obtained:
\begin{align}
z'^2-t'^2-\paren{1/g}^2-\rho^2&=\frac{2}{a_\ret}\left[\paren{1-v_\ret^2}\paren{R_{z'}-R_0v_\ret}-\rho^2 a_\ret\right]\ , \\
\rho z'&=\frac{1}{a_\ret}\left[\paren{1-v_\ret^2}\rho+R_{z'}\rho a_\ret\right]\ , \\
\rho t'&=\frac{1}{a_\ret}\left[\paren{1-v_\ret^2}\rho v_\ret+R_0\rho a_\ret\right]\ ,\\
z'^2-t'^2-\paren{1/g}^2+\rho^2&=\frac{2}{a_\ret}\paren{1-v_\ret^2}\paren{R_{z'}-R_0 v_\ret}\ .
\end{align}
Consequently, the potential and the field components in the inertial frame can be written as
\begin{align}
\phi_{(M)_C}&=ea_\ret\frac{R_{z'}-R_0v_\ret}{\paren{1-v_\ret^2}^{3/2}\paren{R_0-R_{z'}v_\ret}}\nonumber\\
&=eg\frac{R_{z'}-R_0v_\ret}{R_0-R_{z'}v_\ret}\ ,
\end{align}
\begin{align}\label{eq:Ezinert}
E^{z'}&=\mathcal{F}^{z'}_{\ t'}=e\frac{1}{\sqrt{1+\Phi^4}}\frac{\paren{1-v_\ret^2}\paren{R_{z'}-R_0v_\ret}-\rho^2a_\ret}{\paren{R_0-R_{z'}v_\ret}^3}\ , \\
 \label{eq:Erhoinert}
E^{\rho}&=\mathcal{F}^{\rho}_{\ t'}=e\frac{1}{\sqrt{1+\Phi^4}}\frac{\paren{1-v_\ret^2}\rho+R_{z'}\rho a_\ret}{\paren{R_0-R_{z'}v_\ret}^3}\ , \\
 \label{eq:Bphiinert}
B^{\phi}&=\mathcal{F}^{z'}_{\ \rho}=e\frac{1}{\sqrt{1+\Phi^4}}\frac{\paren{1-v_\ret^2}\rho v_\ret+R_0\rho a_\ret}{\paren{R_0-R_{z'}v_\ret}^3}\ .
\end{align}
These equations highly resemble the well-known expressions for the fields of a linearly accelerated charge \cite{Gupta1998}.

The total radiated power per unit solid angle at infinity will be given by\footnote{The emitted power is denoted by $P_0$ to distinguish it from the invariant $P$.}
\be \label{eq:powerstd}
\frac{\dd P_0}{\dd \Omega}=\lim_{\vert\mathbf{x}\vert\to\infty}\vert\mathbf{x}\vert^2\mathbf{S}\cdot\hat{\mathbf{x}} \ ,
\en
Note that, as in Maxwell's case (see for instance \cite{lechner2018}),
the electric and magnetic fields are composed of velocity and acceleration dominating terms. Just as in the Maxwell scenario, at large distances, the field components displayed in Eqs. \eqref{eq:Ezinert}-\eqref{eq:Bphiinert} have a dominant contribution proportional to the acceleration. Hence, the relevant contributions to be used in $\mathbf{S}$ in Eq. \eqref{eq:powerstd}, namely, the acceleration-dominated terms, are
\begin{align}
E^{z'}_\mathrm{ac}&= -\frac{e a_\ret \rho ^2}{\sqrt{1+\Phi^4} (R-R_{z'} v_\ret)^3}, \\
E^{\rho}_{\mathrm{ac}}&=\frac{e a_\ret \rho  R_{z'}}{\sqrt{1+\Phi^4} (R-R_{z'} v_\ret)^3} , \\
B^{\phi}_{\mathrm{ac}}& =\frac{e a_\ret \rho  R}{\sqrt{1+\Phi^4} (R-R_{z'} v_\ret)^3} \ ,
\end{align}
where $\Phi\equiv\phi_{(M)C}/\sqrt{e\beta}$.

The invariant $F$ is given by
\be
F=-\frac{\beta ^2 e^2 \left(1-v_\ret^2\right)^2 \left(R^2 v_\ret^2-2 R R_{z'} v_\ret+R_{z'}^2-\rho ^2 \left(v_\ret^2-1\right)\right)}{(R-R_{z'} v_\ret)^2 \left(e^2 g^4 (R_{z'}-R v_\ret)^4+\beta ^2 (R-R_{z'} v_\ret)^4\right)}.
\en
Let us remark that there is a nontrivial dependence on the acceleration $g$, both in the invariant $F$ and in the potential $\phi_{(M)C}$. 
Also, all the field components
depend nontrivially on $\sqrt{e\beta}$. This will cause the total radiated power to depend on the parameters $(g,e,\beta)$ . 

It will be convenient in what follows to use the dimensionless parameters by resorting to the characteristic nonlinearity radius $\lambda_{\text{BI}}\equiv\sqrt{e/\beta}$,
\be
\mathtt{c}=\frac{e^2g^4}{\beta^2}=\lambda_\text{BI}^4g^4\ ,\quad \texttt{b}=R^4g^4\ .
\en
Note that, since $R$ is the distance from the source to the observation point, and we are interested in the case where $R\to\infty$, the parameter $\texttt{b}$ is bound to be large. Finally, note that on both definitions of $\mathtt{c}$ and $\texttt{b}$ we have lengths, namely $\lambda_\text{BI}$ and $R$, that can eventually be compared.

Let us point out that the nonlinearity radius 
$\lambda_\text{BI}$
defines the region where departures from Maxwell's theory are present, \textit{i.e.,} regions in the nonlinear regime. Since the nonlinearity region 
increases with 
$\lambda_\text{BI}$ (which 
in turn depends of $\beta^{-1/2}$), we may encounter the hypothetical situation of $\lambda_\text{BI}\sim R$, \textit{i.e.,} the case where the nonlinearity region comprises the space up to the asymptotic region. Thus, we can introduce a parameter $\mathtt{a}$, which controls the ratio between the nonlinearity region, $\lambda_\text{BI}$, and the observation point, $R$, given by
\be
\mathtt{a}=R/\lambda_\text{BI}\ .
\en
Consequently, when $\mathtt{a}\gg1$ the observation point is far away from the nonlinearity region and we are expected to recover the results from Maxwell's theory. Finally, we can relate the three parameters $(\mathtt{a},\mathtt{b},\mathtt{c})$ in the form
\be\label{eq:relparams}
\mathtt{b}=\mathtt{a}^4\mathtt{c}\ .
\en
Hence, {by virtue of Eq. \eqref{eq:relparams}, it is possible to reduce the parametric dependence of our results to only two independent parameters. In what follows, we shall consider the parameter space generated by $\paren{\mathtt{a},\mathtt{c}}$.}

For the calculation of the angular distribution of the radiated power, the relevant quantities 
in the limit $\vert\mathbf{x}\vert\to\infty$ are needed.
From Eq. \eqref{eq:powerstd} we obtain that the power emitted per unit solid angle expressed in retarded time
is given by
\be 
\label{eq:powerBI}
\frac{\dd P_0'}{\dd\Omega}=\frac{e^2a_\ret^2\sin^2\theta}{4\pi f(\theta,v_\ret,
\mathtt{a},\mathtt{c})}\ ,
\en
where the angle between $\mathbf{x}$ and $\mathbf{R}$ was neglected by virtue of the limit, and
the function $f$
is given by
\begin{align}
f(\theta,v_\ret,
\mathtt{a},\mathtt{c})=(1-v_\ret \cos\theta)\sqrt{1-\frac{\left(1-v_\ret^2\right)^2}{\mathtt{a}^4 \left(\mathtt{c}  (v_\ret-\cos\theta)^4+(1-v_\ret \cos \theta)^4\right)}}   \nonumber\\
\times\left(\cos ^4\theta \left(\mathtt{c} +v_\ret^4\right)+\mathtt{c}  v_\ret^4-4 \cos\theta \left(\mathtt{c}  v_\ret^3+v_\ret\right)-4 v_\ret \cos ^3\theta \left(\mathtt{c} +v_\ret^2\right)+6 (\mathtt{c} +1)v_\ret^2 \cos ^2\theta +1\right)\ .
\end{align}
{This expression reduces to the one in Maxwell's theory in the limit $\mathtt{a}\to\infty$, $\mathtt{c}\to 0$}.
The nontrivial dependence of the radiated power per solid angle 
on the parameters ($e$, $\beta$, $R$, and $g$) causes departures from the standard results, as will be shown next.  
In fact, for a fixed charge, the angular distribution of radiation in BI electrodynamics is not only dependent on the velocity of the source of the radiation, but also on the value of the maximum field parameter and the acceleration.

Let us begin by inspecting the angle of maximum intensity of radiation, $\theta_{\rm max}$,  by maximizing Eq. \eqref{eq:powerBI}. Figure \ref{fig:thetamax} displays $\theta_{\rm max}$ as a function of {$\mathtt{c}$},
for several values of {$v_\ret$}, and
$\texttt{a}=10^6$, so that the Maxwell limit corresponds to small values of $\mathtt{c}$. The figure shows that for $\mathtt{c}\ll1$ (and $\texttt{a}\gg1$), Maxwell's results for the maximum angle are recovered while, as $\mathtt{c}$ increases, departures from Maxwell theory become considerable: for a given velocity, the maximum angle is larger (for larger $\mathtt{c}$) in comparison with that of Maxwell's. Also, for fixed $v_{\rm{ret}}$ and $\lambda_{\rm{BI}}$, with $\texttt{a}\gg1$, $\theta_{\rm{max}}$ grows with $g$.
\begin{figure}[ht]
\centering
\includegraphics[width=0.47\textwidth]{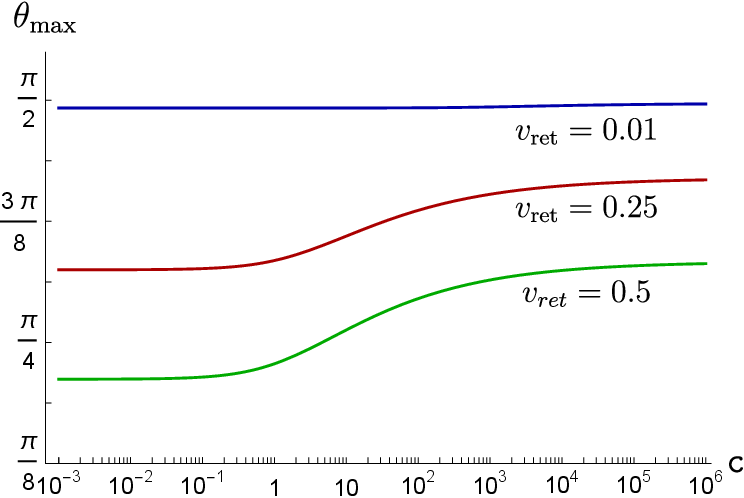}
\caption{Angle of maximum radiation intensity as a function of {$\mathtt{c}$, for $\texttt{a}=10^6$ and $v_\ret=0.01$ (blue curve), $v_\ret=0.25$ (red curve) and $v_\ret=0.5$ (green curve)}. Note that as $\mathtt{c}$ increases, \textit{i.e.}, as the departures from Maxwell's theory are present in a larger region of space, the value of $\theta_\mathrm{max}$ increases.}
\label{fig:thetamax}
\end{figure}

In order to obtain the total energy flux at infinity of the radiation emitted by the accelerating particle, Eq. \eqref{eq:powerBI} must be integrated over the solid angle. Figure \ref{fig:powerBI} displays the normalized flux at infinity as a function of {$\mathtt{c}$}, for $\texttt{a}=10^6$ and several values of $v_\ret$. The classical result (given by Larmor's formula), valid when {$v_{\ret}\ll1$} and $\lambda_{\text{BI}}\to0$, is $P_0\sim 2/3$, and is effectively recovered in the low velocity limit for $\texttt{a}\gg1$ {and small $\mathtt{c}$}. 

It also follows from Figure \ref{fig:powerBI} that when $\mathtt{c}$ increases there are departures from Maxwell's result. In particular, the total power at infinity decreases for increasing $\mathtt{c}$, up to the point that for $\mathtt{c}\gg1$, there is practically no power emitted by the source. This result, is a consequence of the theoretical exploration of the maximum field parameter $\beta$ being arbitrarily small, so that the region $\texttt{a}\ll1$, which can be interpreted as the nonlinearity regime of BI electrodynamics permeating most of space. Such a regime may be called the  deep BI regime, where the field can hardly be excited. 

Since an accelerating charge in the deep BI regime does not lose energy in form of radiation, it should  not react to its own acceleration-varying field, \textit{i.e.}, there should be no  radiation-reaction in such a regime (assuming energy conservation). We shall explore this situation next.    

\begin{figure}[ht]
\centering
\includegraphics[width=0.47\textwidth]{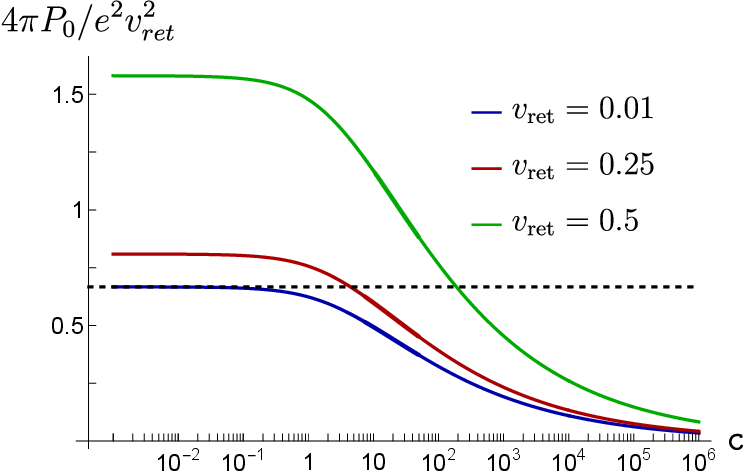}
\caption{Plot of the normalized energy flux at infinity by a BI source (Eq. \eqref{eq:powerBI}) for $\texttt{a}=10^6$ and different values of $v_\ret$. In the plot, $\mathtt{c}$ takes values in the range $[10^{-3},10^6]$. The black dashed line is fixed at $2/3$, which {corresponds to} Maxwell's result for a nonrelativistic particle. Note that for $v_\ret\ll1$ and $\mathtt{c}\ll1$ we recover the power emitted by a Maxwell source.}
\label{fig:powerBI}
\end{figure}

\subsection{Radiation Reaction} \label{sec:radreact}
\begin{figure}[ht]
\centering
\includegraphics[width=0.47\textwidth]{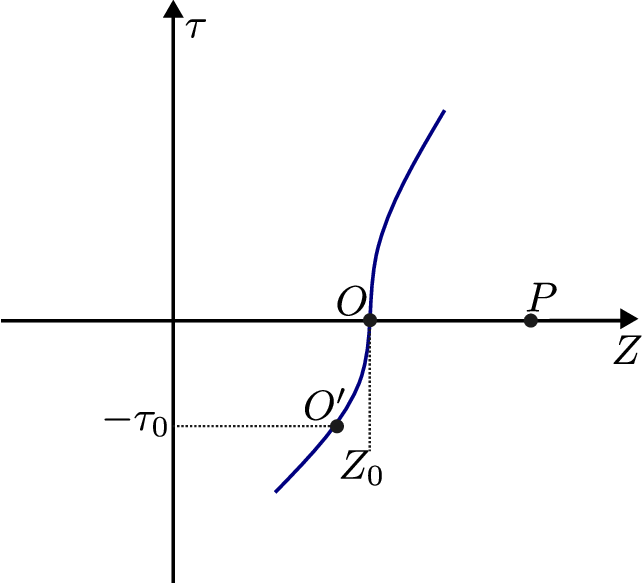}
\caption{Schematic view of the situation considered in Section \ref{sec:radreact}: the blue curve corresponds to the worldline of a BI particle moving with an arbitrary acceleration, in the event $O'$ at $-\tau_0$, such a particle emits radiation that is ultimately detected at point $P$, which is simultaneous with the event $O$. }
\label{fig:fig1}
\end{figure}

In order to examine the problem of how an accelerated BI particle reacts to its radiation field, we shall follow the approach of Gupta and Padmanabhan \cite{Gupta1998}, who start by considering two events $O$ (at $\tau=0$) and $O'$ (at $\tau=-\tau_0$) on the worldline of the accelerated particle\footnote{Let us remark that this approach is valid for a particle under an arbitrary acceleration, not necessarily constant.}. These arbitrarily chosen events represent spacetime points where the uniformly accelerated particle is instantaneously at rest and a point over the past worldline, respectively. {Figure \ref{fig:fig1} schematically represents this situation {in the instantaneously-at-rest coordinates $(\tau,Z)$}. We shall return to it below.} Although the reference frame $S'$, introduced in Sec. \ref{sec:inertcomp}, is comoving with the accelerating particle (hence it is a non-inertial frame), 
in a coordinate system at rest at any event (in particular at $O$), the particle is instantaneously at rest. Furthermore, at $O$, it follows from the $\rho\to0$ limit in
Eq. \eqref{eq:Erho} (\textit{i.e.}, when taking points over the particle's worldline), $E^{\rho}$ vanishes identically, while $B^\phi$ is trivially zero because the field is evaluated at the particle's rest frame. 

In order to write the near-horizon line element in a functional form 
analog to that of Rindler's spacetime, 
the coordinate transformation $\lambda=\kappa_s Z^2/2$ can be performed in Eq. \eqref{eq:dslambda}. In the $\paren{t,Z,\rho,\varphi}$ coordinate system the line element is
\be \label{eq:dsapproxrho}
\dd s^2=\kappa_s^2Z^2\dd t^2-\dd Z^2-\dd\rho^2-\rho^2\dd\varphi^2\ ,
\en
and Copson's potential reads
\be\label{eq:VCZ}
\phi_{M(C)}\paren{Z}=\frac{e}{2r_s}\frac{Z^2+\rho^2+\paren{1/\kappa_s}^2}{\sqrt{\paren{Z^2+\rho^2-\paren{1/\kappa_s}^2}^2+4\rho^2\paren{1/\kappa_s}^2}}\ .
\en
The expression for the $Z$ component {of the electric field} is
 \be
E^{Z}=-\frac{1}{gZ}\frac{\p_Z \phi_{M(C)}\paren{Z}}{\sqrt{1+\vert\boldsymbol{\nabla} \phi_{M(C)}\paren{Z}\vert^2\beta^{-2}}}\ ,
\en
where {$\phi_{M(C)}\paren{Z}$} is
given by Eq. \eqref{eq:VCZ}. Note that the invariant $P$ is not present in this equation, since we are only considering the limit on the particle's position. Explicitly, as $\rho\to0$, we get
\be \label{eq:EZapp1}
E^{Z}=\frac{4 e  }{g^2\left(Z^2-Z_0^2\right)^2 \sqrt{\frac{16 e^2 Z^2 }{g^2\beta ^2 \left(Z^2-Z_0^2\right)^4}+1}}\ ,
\en
where $Z_0$ is the position of the particle. As expected, {in the case of} Maxwell electrodynamics namely, in the $\beta\to\infty$ limit, in which the potential is $\phi_{M(C)}(Z)$, there is a divergence on the particle's position which scales as {$Z^{-2}\sim\lambda^{-1}$}. For any finite value of $\beta$,  Eq. \eqref{eq:EZapp1} is regular when $Z=Z_0$

Let us analyze how the motion of the particle affects its neighborhood by considering {the situation schematically presented in Figure \ref{fig:fig1}.} The goal is to obtain the field at point $P$, which is simultaneous with $O$ and 
on the future light cone of $O'$. In order to consider an infinitesimal neighborhood around $O$, we take $P\to O$ and $O'\to O$. In the limit $O'\to O$, the particle is instantaneously at rest, and the electric field is given by Eq. \eqref{eq:EZapp1}. 
By setting $\rho=0$ in \eqref{eq:75} and expressing the result in $(t,Z,\rho,\phi)$
coordinates we find that the relation
\be\label{eq:idzvmu}
Z^2-Z_0^2=\frac{2}{g}R_\mu v_{\ret}^{\mu}\ ,
\en
holds. Hence, $E^Z$ in Eq. \eqref{eq:EZapp1} can be written in terms of $R_\mu v^\mu_\ret$. 

When taking the limit $O'\to O$, the dominant contribution is that of the static field -- in this case, the field of a BI point particle -- and radiation reaction arises when considering the particle as non-static. {Hence, the motion of the particle at $O'$ must be taken into account, by expanding  $R_\mu v^\mu_\ret$
as follows:}
\begin{align} 
&R_\mu v_{\ret}^{\mu}\vert_{O'}=R_\mu v_{\ret}^\mu\vert_O+\p_Z\left(R_\mu v_{\ret}^\mu\right)\paren{Z-Z_0} ,
\end{align}
where the leading order term is given by Eq. \eqref{eq:idzvmu}. 
It follows that
\begin{align} 
\label{eq:Rmuvmu}
&R_\mu v_{\ret}^\mu=\frac{g}{2}\left(Z^2-Z_0^2(-\tau_0)\right)-\left(Z-Z_0(-\tau_0)\right)u^{Z}\paren{-
\tau_0}\ ,
\end{align}
where $u^Z\paren{-\tau_0}$ is the velocity of the particle at $O'$. Hence, the new expression for the $E^Z$ component of the field 
follows from replacing Eqs.  \eqref{eq:idzvmu} and \eqref{eq:Rmuvmu} into Eq. \eqref{eq:EZapp1}.

Next{, in the instantaneously-at-rest frame}, at the event $O$, the particle has zero velocity and acceleration, but the derivative of the acceleration at such a point is nonzero. We can expand the position $Z_0\paren{-\tau_0}$ for $O\to O'$ as
\be
Z_0(-\tau_0)\approx Z_0-\frac{1}{6}\dot{\alpha}\tau_0^3\ ,
\en
where $Z_0=1/g$ and  $\dot{\alpha}\equiv\dd^3 Z_0/\dd\tau^3$ is the derivative of the proper acceleration, which is zero for the case of hyperbolic motion. Similarly, we obtain for the velocity
\be
u^Z\paren{-\tau_0}\approx\frac{1}{2}\tau_0^2\dot{\alpha}\ .
\en
Taking into account that the observation point $P$ is very close to the particle's position, we can write $Z=Z_0+\delta$ with $\delta/Z_0\ll1$. Finally we arrive to an expression for the electric field that depends of two parameters, $\tau_0$ and $\delta$. It follows from Eqs. \eqref{eq:74} and \eqref{eq:75}  that, by setting $\tau=\rho=0$ for small $\tau_0$ we obtain 
\be
\delta\sim\tau_0\ .
\en
{Using these relations in the expression for $E^Z$}, the force on the particle exerted by its own field, $F^Z=eE^Z(P)$, can be expanded as $\delta\to0$, for fixed beta, yielding
\begin{align}\label{eq:ForceFiniteBeta}
F^Z\approx& e\beta-e g\beta\delta+eg^2\beta\delta^2-eg^3\beta\delta^3+\paren{eg^4\beta-\frac{\beta^3}{2e}}\delta^4+\paren{-eg^5\beta+\frac{g\beta^3}{2e}}\delta^5+\paren{eg^6\beta-\left[\frac{3}{4}g^2-\frac{2}{3}\dot{\alpha}\right]\frac{\beta^3}{e}}\delta^6+\mathcal{O}\paren{\delta^7}\ .
\end{align}
The series has been explicitly calculated 
up to the term where the derivative of the proper acceleration, $\dot{\alpha}$, appears.

Let us now discuss the result \eqref{eq:ForceFiniteBeta}. Since BI electrodynamics was specifically engineered to furnish regular values for both the electric field and the energy of a point particle, it is expected that the force exerted by the particle's field on the particle should also be regular. Eq. \eqref{eq:ForceFiniteBeta} shows that this is indeed the case. The expansion in 
Eq. \eqref{eq:ForceFiniteBeta} goes up to $\mathcal{O}(\delta^6)$ in order to show at what order of the expansion
the motion of the particle starts to influence the force. {Interestingly,} the derivative of the proper acceleration appears only {at} order 6, while in ME it appears at order zero, that is, as a finite contribution to the force (see \ref{app:Gupta}). Hence
its effect is negligible in the case under study, for any finite $\beta$. Also, it is interesting to note that the numerical coefficient appearing in the {term carrying $\dot{\alpha}$} is precisely the same that appears in Maxwell's result  (see Eq. \eqref{eq:radreact}). 

We shall show next how the result for the 
radiation-reaction force in Maxwell's theory, \textit{i.e.}, two divergent terms and some constant terms follows from Eq. \eqref{eq:ForceFiniteBeta}. It is easily seen that taking the limit in which the region of nonlinearity goes to zero (namely $\beta\to\infty$) does not yield the correct Maxwellian limit. It also follows that 
the imposition of the equality $\beta=e/\delta^2$ leads to the Maxwellian expression for the force from Eq. \eqref{eq:ForceFiniteBeta}. However, it is important to point out that the development in Eq. \eqref{eq:ForceFiniteBeta} goes up to $\mathcal{O}(\delta^6)$. Hence, by assuming that $\beta\propto\delta^{-2}$, we may be neglecting higher order combinations of the type $\beta^n\delta^{2n}$, with $n\in\mathbb{N}$, that contribute non-trivially to the regular terms in the force. Therefore, in order to obtain Maxwell's result for the force, we need to independently consider the limits $\lambda_{\text{BI}}\to0$ and $\delta\to0$, thus obtaining the result 
\be
F^{Z}\approx\frac{e^2}{\delta ^2}-\frac{e^2 g}{\delta }+\frac{1}{12} e^2 \left(8 \dot{\alpha} +9 g^2\right)+\mathcal{O}(\delta,\lambda_{\text{BI}}^4)\ ,
\en
\textit{i.e.}, the exact form of the radiation-reaction force in Maxwell electrodynamics {c.f. Eq. \eqref{eq:radreact}}, in the instantaneously-at-rest frame\footnote{See Eq. (D.6) for the expression in the inertial coordinates.}
.
\begin{figure}[ht]
\centering
\includegraphics[width=0.42\textwidth]{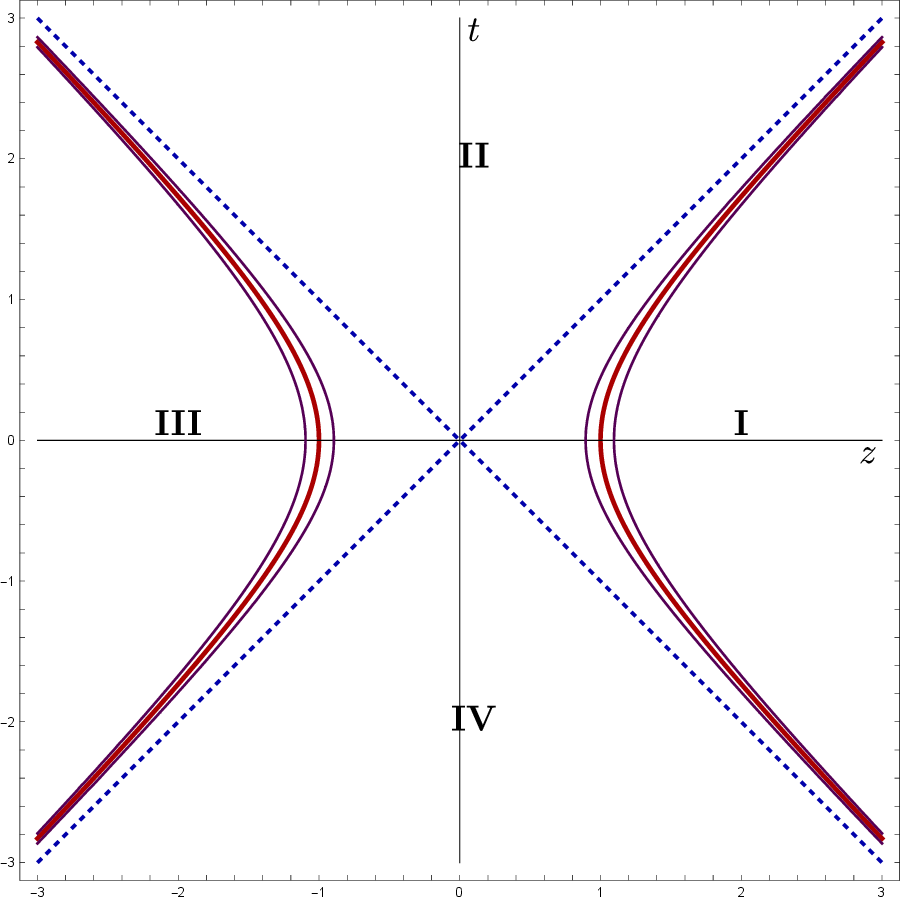}
\caption{The nonlinearity region,  modelled as $\vert\nabla \phi_{M(C)}\vert\approx\beta$, corresponds to the purple lines in the above plot for $\beta=0.5$ and $e=0.01$. The red line corresponds to the world line of an accelerated particle.}
\label{fig:BI_reg_Rindler}
\end{figure}

The nonlinearity region, defined by $\lambda_{\text{BI}}$, where  the electric field is saturated, can be defined by $\vert\nabla \phi_{M(C)}\vert=\beta$ (see Eq. \eqref{eq:normEBI}) 
as displayed in Fig. \ref{fig:BI_reg_Rindler}. Consequently, in order to account for the nonlinearity contributions 
to the force we can consider three different situations, namely \textit{(i)}  $\lambda_{\text{BI}}<\delta$, \textit{i.e.}, the observation point at $P$ is outside the nonlinearity region, \textit{(ii)} $\lambda_{\text{BI}}\sim\delta$, the observation point is located at the limit of the nonlinearity region and \textit{(iii)} $\lambda_{\text{BI}}>\delta$, the observation point is enclosed by the nonlinearity region. In particular, we may analytically take the limit $\delta\sim \lambda_{\text{BI}}$ in Eq. \eqref{eq:ForceFiniteBeta},  
obtaining, at leading order,
\be \label{eq:FZr0delta}
F^{Z}=\frac{e^2}{\sqrt{2} \delta ^2}-\frac{e^2 g}{\sqrt{2} \delta }+\frac{e^2 \left(8 \dot{\alpha} +21 g^2\right)}{24 \sqrt{2}}+\mathcal{O}(\delta)\ ,
\en
a result which displays the same behavior of Maxwell's result for  both the divergent and regular terms, aside {from} numerical factors. 

Note that the divergent behavior present in Eq. \eqref{eq:FZr0delta} is a consequence of the fact
that we are taking the nonlinearity radius as vanishing with $\delta$.

Let us stress that, just as in the Maxwell case, terms in
Eq. \eqref{eq:ForceFiniteBeta}
that are characteristic of the electric field on the particle itself, \textit{i.e.}, terms that arise even when considering the static point particle, 
are present 
in the BI case. Therefore, the interesting terms in the expression of the force are those that arise exclusively from the motion of the particle, namely, those that include the derivative of the proper acceleration and are ultimately related to the radiation-reaction force\footnote{{The appearance of the derivative of the proper acceleration, coming from the consideration of the motion of the particle at $O'$ is related to the fact that radiation-reaction leads to a deviation in the particle's trajectory. In fact, this effect has a purely geometrical origin (see \cite{Gupta1998}).}}. Then, it is suitable to consider only such terms by subtracting the static part of the force, \textit{i.e.}, by taking
\be\label{eq:FZregular}
F^Z(\dot{\alpha})=F^Z(P)-F^Z(O)\ ,
\en
where $F^Z(P)$ is the force calculated at point $P$ as $P\to O$ as previously stated, and $F^Z(O)$ is the force at event $O$ (at the origin of the comoving system). 
Figure \ref{fig:FZreg} displays the dependence of $F^Z$ with $\dot{\alpha}$ for the situations \textit{(i)}, \textit{(ii)}, and 
\textit{(iii)} described above, for fixed $\delta$. In all cases, $F^Z$ depends linearly of  $\dot{\alpha}$, and the radiation-reaction term attains a maximum when $\lambda_{\text{BI}}=\delta=0$, i.e., in Maxwell's electrodynamics. In any other case, the numerical factors multiplying the radiation reaction term are smaller than those in Maxwell's electrodynamics. \\
It is important to stress that if the observation point lies outside the nonlinearity region
($\delta>\lambda_{\text{BI}}$), then the radiation-reaction force is stronger than in the case 
$\delta<\lambda_{\text{BI}}$. This is related to the fact that the nonlinearity region is characterized as the region where the field takes its maximum possible value, $\beta$. Hence, it is impossible to increase the force in this region further. The nonlinearity region acts like a screening mechanism for the particle's electric field. Consequently, within this region, it is impossible for radiation-reaction to arise. Another interesting situation comes from considering the observation point outside the nonlinearity region; Maxwell's result is adequately recovered if and only if this region  shrinks into a point. Otherwise, the radiation-reaction term is always less than in Maxwell's electrodynamics. 
\begin{figure}[ht]
\centering
\includegraphics[width=0.47\textwidth]{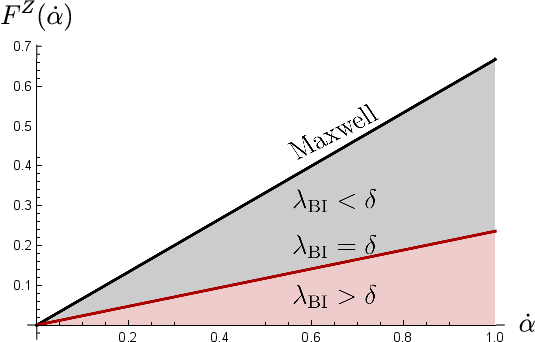}
\caption{
The plot shows the part of $F^Z$ 
that depends of the motion of the particle
as a function of the
derivative of the proper acceleration. The black line corresponds to Maxwell's result, while the red line corresponds to the case where $\lambda_{\text{BI}}=\delta$. As the nonlinearity region increases, $F^Z$ decreases up to the point where it is negligible. We have chosen $e=1$, $g=1$, $\delta=10^{-4}$. }
\label{fig:FZreg}
\end{figure}

Similarly, let us we consider the situation where the nonlinearity region is arbitrarily large. It can be seen directly from Eq. \eqref{eq:ForceFiniteBeta} that the radiation-reaction term vanishes since $\delta$ is assumed to be a very small parameter close to the position of the classical charge. This should eventually have consequences related to energy conservation since radiation-reaction for Maxwell's theory arises as a manifestation that a part of the particle's electric field is responsible for a non-zero power at large distances from the source, implying that this energy loss must be compensated in regions close to the particle, which coincides with the {results obtained in the previous section}.

\section{Conclusions}\label{sec:Conclusions}

Physically meaningful solutions of nonlinear electromagnetic theories are hard to obtain, even in static cases with a high degree of symmetry. The addition of charges in motion makes matters worse. In BI electrodynamics, the dynamics of a charged particle is usually analyzed using a perturbative approach \cite{Chruscinski1998,Chruscinski1998a}. In the present work, we have analyzed the behavior of a linearly accelerating charged BI particle using the electrostatic potential of a BI charged particle in the vicinity of a Schwarzschild black hole \cite{Falciano2019}. By exploiting the fact that the spacetime arbitrarily close to the horizon of a Schwarzschild black hole is reduced to Rindler's spacetime, and using the theorem presented in Sec. \ref{ssec:BIpart}, we calculated the electrostatic potential and fields in the comoving system. These quantities were then transformed to the inertial frame. The electromagnetic field components for such a particle were also calculated.

Armed with these results, we have discussed the influence of the nonlinearities on the radiation. In doing so, we have explored the theoretical situation of $\beta$ being a constant whose value determines only the region where the nonlinearities of the theory become important (\textit{i.e.}, where Maxwell's results do not hold) and, as such, has not a fixed value. In this context, we have shown that both the radiation field and the radiation-reaction force are influenced by $\beta$.
In fact, for small values of this parameter, both  radiation and radiation-reaction for the accelerated particle are less important than those in Maxwell's case, while 
the nonlinearities are still relevant for physically viable values for the $\beta$ parameter (\textit{i.e.}, when $\beta\gg1$).
The effects of the influence of nonlinearity on radiation and radiation-reaction are discussed in Sec. \ref{sec:Radiation}, and are synthesized in Figures \ref{fig:thetamax}, \ref{fig:powerBI} and \ref{fig:FZreg}. These features can be traced back to the existence of the nonlinearity region defined by $\lambda_{\text{BI}}=\sqrt{e/\beta}$, which offers a screening mechanism for the electromagnetic field. It becomes increasingly easier to saturate the electric field within this region. Therefore, in the case where the nonlinear regime dominates in a very large region, the accelerated BI particle emits less radiation than in the opposite case, which in turn results in the weaker response of the particle to this radiation. We also showed how the $\beta$ parameter of Born-Infeld electrodynamics influences both the angle of maximum radiation intensity $\theta_{\text{max}}$ and the total radiated power for the uniformly accelerated charge. In particular, the more extended the departure from Maxwell’s theory is in space, the larger the value of $\theta_{\rm{max}}$. 

Let us remark that an important part of the present work resides in showing that, despite the general impossibility of having a well-defined electrostatic potential for NLED, it is still possible to derive approximated results for this potential, which lead to the electric field of an accelerating particle at a  fair level of approximation, specially in the neighbourhood of the particle. Also, we must stress that, contrary to what was previously obtained in the literature for ME, our result for BI furnishes only regular contributions to the particle's force, as expected, and the Maxwellian radiation-reaction force is recovered in the proper limit. This last statement is a consequence of the fact that the expression of the Lorentz force was used in the calculations. Additional research using alternative, well-motivated definitions of force (c.f. \cite{Rohrlich2008}) may is also of interest. 

Let us also mention that we presented evidence that indicates that the electric field of a Born-Infeld particle vanishes on the horizon of a Schwarzschild black hole (a result that was known to be valid for Maxwell's theory, in which the horizon acts like a conducting surface). An immediate consequence of this fact is that the black hole polarization phenomenon is also present within BI electrodynamics. Whether this is also true for other nonlinear electromagnetic theories remains to be seen. 

Finally, let us stress that the method introduced in this work can, in principle, be applied to any electromagnetic theory, thus paving the way to the study of the radiation field and radiation-reaction in other nonlinear theories of interest. We hope to return to this problem, as well as to that of the alternative definitions of force, in future publications.

\section*{Acknowledgements}
The authors thank the anonymous referee for useful suggestions towards the improvement of a previous version of the manuscript. This work was supported by National Council for Scientific and Technological Development - CNPq and FAPERJ - Fundaçāo Carlos Chagas Filho de Amparo à Pesquisa do Estado do Rio de Janeiro, Processo SEI 260003/014960/2023.

\appendix

\section{Copson and Linet's potential in Rindler coordinates} \label{app:PotRind}

The Schwarzschild metric in {standard} coordinates is given by 
\be \label{eq:SmetricA}
\dd s^2=\left(1-\frac{r_s}{r}\right)\dd t^2-\frac{\dd r^2}{\left(1-\frac{r_s}{r}\right)}-r^2\dd\Omega^2\ ,
\en
where $r_s=2M$ is the horizon radius. To obtain the near-horizon metric, we can expand the radial coordinate as $r\to r_s\paren{1+\epsilon\tilde{r}/r_s}$ with $\epsilon\ll1$. Retaining only the first order in $\epsilon$, we obtain
\be \label{eq:1stcoord}
\dd s^2=\epsilon\frac{\tilde{r}}{r_s}\dd t^2-\epsilon\frac{r_s}{\tilde{r}}\dd\tilde{r}^2-r^2\dd\Omega^2
\en
Considering  
$r\approx r_s,\theta\approx0$, we have that
\begin{eqnarray*}
\rho&=&r\sin\theta \\ 
&=&r_s\paren{1+\epsilon\frac{\tilde{r}}{r_s}}\sin\theta \\
&=&r_s\theta+\mathcal{O}\paren{\theta^2,\epsilon^2,\epsilon\theta}\ ,\\
\Rightarrow \dd \rho&=&r_s\dd\theta\ .
\end{eqnarray*}
Then, the line element reads
\be
\dd s^2=\epsilon\frac{\tilde{r}}{r_s}\dd t^2-\epsilon\frac{r_s}{\tilde{r}}\dd\tilde{r}^2-\dd\rho^2-\rho^2\dd\varphi^2\ .
\en
Using the surface gravity of a Schwarzschild black hole, $\kappa_s=1/2r_s$, and defining 
$\lambda=\epsilon \rtil$, the metric takes the form
\be \label{eq:apdslambda}
\dd s^2=2\kappa_s\lambda\dd t^2-\frac{1}{2\kappa_s\lambda}\dd\lambda^2-\dd\rho^2-\rho^2\dd\varphi^2\ ,
\en
which is extensively used throughout the paper. In these coordinates,  the metric reduces to that of Minkowski spacetime when $\lambda=\lambda_0=1/2\kappa_s=r_s$. Then, in this coordinate system the particle is at rest when $\lambda=\lambda_0$.

We can also use the coordinate system $\paren{t,Z,\rho,\varphi}$ 
defined by the transformation $\lambda=\kappa_s Z^2/2$, thus obtaining 
\be \label{eq:dsapproxrho1}
\dd s^2=\kappa_s^2Z^2\dd t^2-\dd Z^2-\dd\rho^2-\rho^2\dd\varphi^2\ .
\en
where, since $\kappa_s$ has dimensions of 1/length, the quantity $\kappa_s^2Z^2$ is dimensionless and in these coordinates the particle is located at $Z_0=2r_s$. By shifting the $Z$ coordinate as $\tilde{Z}=Z-1/\kappa_s$, 
it follows that
\be \label{eq:dsapproxZ}
ds^2=\paren{1+\kappa_s\tilde{Z}}^2\dd t^2-\dd\tilde{Z}^2-\dd\rho^2-\rho^2\dd\varphi^2\ .
\en
It is particularly convenient to use either Eq. \eqref{eq:dsapproxrho1} or Eq. \eqref{eq:dsapproxZ} because the spatial part of the line element can be mapped directly into usual Cartesian coordinates.

We now proceed to calculate Copson and Linet's potential in near-horizon coordinates. In standard coordinates,
Copson's potential for a particle located at a distance $a$ from the origin is given by
\be
\phi_{M(C)}=\frac{e}{2ar}\frac{\paren{2r-r_s}\paren{2a-r_s}-r_s^2\cos\theta}{\left[\paren{2r-r_s}^2+\paren{2a-r_s}^2-2\paren{2a-r_s}\paren{2r-r_s}\cos\theta-r_s^2\sin^2\theta\right]^{1/2}}\ .
\en
By performing the coordinate transformation displayed in Eq. \eqref{eq:1stcoord} we obtain
\be \label{eq:apVcepsilon}
\phi_{M(C)}\paren{\rtil}
=\frac{e}{2r_s}\frac{\epsilon\atil+\epsilon\rtil+\rho^2\kappa_s/2}{\sqrt{\paren{\epsilon\left[\rtil+\atil\right]+\kappa_s\rho^2/2}^2-2\epsilon\kappa_s\rho^2\atil}}\ .
\en
Let us now transform the potential in Eq. \eqref{eq:apVcepsilon} to the other coordinate systems presented in this section. In the $(t,\lambda,\rho,\varphi)$ coordinate system {we have} 
\be \label{eq:ApVc}
\phi_{M(C)}\paren{\lambda}=\frac{e}{2r_s}\frac{\lambda+\lambda_0+\rho^2\kappa_s/2}{\sqrt{\left[\lambda-\lambda_0+\rho^2\kappa_s/2\right]^2+2\kappa_s\lambda_0\rho^2}}
\en
Transforming to the $(t,Z,\rho,\varphi)$ coordinate system, 
\be \label{eq:ApA0}
\phi_{M(C)}\paren{Z}=\frac{e}{2r_s}\frac{Z^2+\rho^2+\paren{1/\kappa_s}^2}{\sqrt{\paren{Z^2+\rho^2-\paren{1/\kappa_s}^2}^2+4\rho^2\paren{1/\kappa_s}^2}}\ .
\en
Defining $\omega\equiv\sqrt{\left[\lambda-\lambda_0+\rho^2\kappa_s/2\right]^2+2\kappa_s\lambda_0\rho^2}$,
the derivatives of the potential defined in Eq. \eqref{eq:ApVc} read
\begin{align}
\p_\lambda \phi_{M(C)}\paren{\lambda}&=-\frac{e \lambda_0 }{r_s\omega^3}\left[\lambda-\lambda_0-\rho^2\kappa_s/2\right]\ , \\
\p_\rho \phi_{M(C)}\paren{\lambda}&=-2\frac{e}{r_s}\frac{\lambda\lambda_0\rho\kappa_s}{\omega^3}\ .
\end{align}
Hence, the norm of the electric field is 
\be
\vert \boldsymbol{\nabla} \phi_{M(C)}\paren{\lambda}\vert^2=\frac{e^2}{\omega^6}\frac{\lambda}{\lambda_0}\left\{\left[\lambda-\lambda_0-\rho^2\kappa_s/2\right]^2+2\lambda\rho^2\kappa_s\right\}\ .
\en
Correspondingly, in the $(t,Z,\rho,\varphi)$ coordinate system, 
\begin{align}
\p_Z \phi_{M(C)}\paren{Z}&=-2\frac{eZ\paren{1/\kappa_s}^2}{r_s\omega^3}\paren{Z^2-\rho^2-\paren{1/\kappa_s}^2}\ , \\
\p_\rho \phi_{M(C)}\paren{Z}&=-4\frac{e}{r_s}\frac{\rho Z^2}{\omega^3}\paren{1/\kappa_s}^2\ ,
\end{align}
and the 
components of the 
electric field in this coordinate system read
\begin{align}
E^{Z}&=4e\paren{\frac{1}{g}}^2\frac{Z^2-\rho^2-\paren{1/\kappa_s}^2}{\omega^3}\ , \\
E^{\rho}&=\frac{8e\rho Z}{\omega^3}\paren{\frac{1}{g}}^2\ .
\end{align}

\section{Electrostatic potential of a particle near the horizon of a Schwarzschild black hole}\label{app:BISbh}

In the case in which the BI particle is arbitrarily close to the horizon of a Schwarzschild black hole, the metric reduces to the form displayed in Eq. \eqref{eq:apdslambda} and the electrostatic potential, given by Eq. \eqref{eq:ApVc}, is a function of the coordinates $(\lambda,\rho)$. Expanding $\phi_{M(C)}$ 
about the position of the particle, \textit{i.e.}, taking $\lambda\to\lambda_0(1+\varepsilon)$ and $\rho\to\delta$, with $\varepsilon\ll1$ and $\delta/\rho\ll1$, we obtain, at leading order 
\be\label{eq:apCopsonapprox}
\phi_{M(C)}\approx\frac{e}{\lambda_0\varepsilon}+\mathcal{O}\paren{\varepsilon^2,\delta^2}\ ,
\en
while the derivatives take the form
\begin{align}
\p_\lambda \phi_{M(C)}&\approx-\frac{e}{\lambda_0^2 \epsilon ^2}+\mathcal{O}\paren{\epsilon^2,\delta^2}\ , \\
\p_\rho \phi_{M(C)}&\approx\mathcal{O}(\delta)\ , \\
\Rightarrow\vert\boldsymbol{\nabla}\phi_{M(C)}\vert&\approx\frac{e}{\lambda_0^2\epsilon^2}+\mathcal{O}\paren{\epsilon^2,\delta^2}\ . \label{eq:apnomVc}
\end{align}

Note that the potential $\phi_{M(C)}$ in Eq. \eqref{eq:apCopsonapprox} represents the electrostatic potential for a charged particle without accounting for the effects of the presence of the black hole which 
would be represented by black hole polarization. In other words, if we were to take $\phi_{M(C)}$ as the complete solution for the potential then the physical situation represented will be that of a uniformly accelerated particle in Rindler spacetime, which is precisely the focus of the present work. Therefore, the electric field can be expressed in terms of the gradient of $\phi_{M(C)}$ as
\be 
\label{eq:apEfieldVc}
\mathbf{E}\paren{\lambda,\rho}=-\boldsymbol{\nabla}\phi=-\frac{\boldsymbol{\nabla}
\phi_{M(C)}}{\sqrt{1+\vert \boldsymbol{\nabla}\phi_{M(C)}\vert^2\beta^{-2}}}\ .
\en
It remains to integrate Eq. \eqref{eq:apEfieldVc} in order to obtain the electrostatic potential $\phi$. Using Eq. \eqref{eq:apnomVc}, we can approximate
\be
\vert\boldsymbol{\nabla}\phi_{M(C)}\vert^2\approx\beta^2\paren{\frac{\phi_{M(C)}}{\Delta}}^4\ ,
\en
where $\Delta\equiv\sqrt{e\beta}$. Then, Eq. \eqref{eq:apEfieldVc} reads
\be
\boldsymbol{\nabla}\phi\approx\Delta\frac{\boldsymbol{\nabla}\paren{\phi_{M(C)}/\Delta}}{\sqrt{1+\paren{\phi_{M(C)}/\Delta}^4}}\ ,
\en
which can be directly integrated to obtain
\be\label{eq:apphi1}
\phi=-\frac{\Delta^2}{\phi_{M(C)}}\ _2F_1\left[\frac{1}{2},\frac{1}{4},\frac{5}{4};-\paren{\frac{\Delta}{\phi_{M(C)}}}^4\right]+C_{\infty}\ .
\en
The integration constant  $C_\infty$ can be obtained by demanding that far from the particle, the potential $\phi$ coincides with Copson's potential $\phi_{M(C)}$. By taking the limit of Eq. \eqref{eq:apphi1} for $\lambda\to\infty$ we obtain
\be
\lim_{\lambda\to\infty}\phi=C_\infty-\frac{\Delta\Gamma\paren{\frac{1}{4}}^2}{4\sqrt{\pi}}+\phi_{M(C)}+\mathcal{O}\paren{\phi_{M(C)}^5}\ .
\en
Hence, $C_\infty=\Delta\Gamma\paren{\frac{1}{4}}^2/4\sqrt{\pi}$ and the approximated electrostatic potential reads
\be
\phi=\Delta\paren{\frac{\Gamma\paren{\frac{1}{4}}^2}{4\sqrt{\pi}}-\frac{\Delta}{\phi_{M(C)}}\ _2F_1\left[\frac{1}{2},\frac{1}{4},\frac{5}{4};-\paren{\frac{\Delta}{\phi_{M(C)}}}^4\right]}\ .
\en

\section{Li\'{e}nard-Wiechert potential and field} \label{app:LWpot}

In order to set up our notation, let us review the derivation of the Liénard-Wiechert (LW) potential. This topic is often addressed in standard textbooks (c.f. \cite{lechner2018,Kosyakov2007}). The LW potential for a particle with charge $e$ is given by
\be
A^{\mu}=e\frac{u^\mu}{(R^\lambda u_\lambda)}\ ,
\en
where $R^\mu=x^\mu-z^\mu(s)$ is the difference between the observation point $x^\mu$ and the point where the particle is located along the worldline $z^\mu(s)$, and $u^\mu$ is the particle's four-velocity.
Note that the relation $R_\lambda R^\lambda=0$ holds, and that the requirement of considering the retarded Green's function can be recast as $R^0>0$. The Faraday tensor for the {LW} potential (often referred to as the LW field) is given by
\be \label{eq:FmunuLW}
F^{\mu\nu}=\frac{e}{\rho^3}\left[\left(R^\mu u^\nu-R^\nu u^\mu\right)\left(1-a^\lambda R_\lambda\right)+\rho\left(R^\mu a^\nu-R^\nu a^\mu\right)\right]\ , 
\en
where $\rho\equiv R_\lambda u^\lambda$, {and $a^\mu$ is the particle's four-acceleration}. From 
Eq. \eqref{eq:FmunuLW}
the invariants $F$ and $G$
can be directly calculated, yielding
\be
F=-\frac{e^2}{\rho^4}\ ; \qquad G=0\ .
\en
The expression for the {LW} field depends of the particle's velocity and acceleration. Hence, 
$F^{\mu\nu}$ can be rewritten as
\be
F^{\mu\nu}=F^{\mu\nu}_v+F^{\mu\nu}_{\mathrm{ac}}\ ,
\en
where
\begin{eqnarray}
F^{\mu\nu}_v&=&\frac{e}{\rho^3}\left(R^\mu u^\nu-R^\nu u^\mu\right)\ , \\
F^{\mu\nu}_\mathrm{ac}&=&\frac{e}{\rho^3}\left[\rho\left(R^\mu a^\nu-R^\nu a^\mu\right)-R^\lambda a_\lambda\left(R^\mu u^\nu-R^\nu u^\mu\right)\right]
\end{eqnarray}
respectively. It is straightforward to show that the terms involving the acceleration of the particle do not contribute to the invariant $F$. The fact that the  acceleration of the particle -- despite being nontrivial -- does not enter in the expression for the invariant $F$ is closely related to the fact that $F$ is a Lorentz invariant and, since inertial observers must agree in the value of such an invariant, it should not depend on the acceleration of the particle. Otherwise, there may exist inertial observers detecting (or failing to detect) phenomena such as radiation from the accelerated particle which, in fact, leads to a series of apparent paradoxes that have attracted attention over the years \cite{Boulware1980,Leaute1983}.

\section{Gupta \& Padmanabhan calculation of the radiation-reaction force.}\label{app:Gupta}

We reproduce here some important results of Gupta \& Padmanabhan's paper regarding the calculation of the radiation reaction force 
in Maxwell's theory \cite{Gupta1998}.  First, in the limit over the particle's position the electric field components behave as (here we are using the $\paren{t,Z,\rho,\phi}$ coordinate system)
\begin{align}
E^Z&\approx\frac{4e}{g^2}\frac{1}{\paren{Z^2-Z_0^2}^2}\ , \\
E^\rho&=0\ .
\end{align}
Clearly, as we approach the particle's position the $E^Z$ component diverges. Moreover, we can express the latter result as 
\be \label{eq:OEField}
E^Z=\frac{e}{\paren{R_\mu v^\mu_{\ret}}^2}\ .
\en
Since we are interested in the case of a uniformly accelerated charge, we need to take into account two events, namely $O$ and $O'$, where the particle is instantaneously at rest (i.e., both the velocity and acceleration at this point are zero) and where the particle is accelerated . Furthermore, we take an observation point $P$ which is simultaneous with $O$. The light cone connecting $O'$ with $P$
is not necessarily a straight line, however since we need to consider points very close to the event $O$, the actual shape of this line is not relevant. 

As $O'\to O$ the behavior of the field will be mainly dominated by the Couloumbian part. At $O$ the field is static and given by Eq. \eqref{eq:OEField}. Since we are considering not points at $O$ but rather very close to it, we need to consider a modification to the denominator in Eq. \eqref{eq:OEField} due to the velocity $u^Z$ of the particle as
$$R_\mu v_{\ret}^{\mu}\approx\frac{g}{2}\left(Z^2-Z_0^2(-\tau_0)\right)-\left(Z-Z_0(-\tau_0)\right)u^Z(-\tau_0)\ .$$
At the event $O$, the particle has zero velocity and acceleration, but the derivative of the proper acceleration 
(denoted by $\dot{\alpha}$)
is nonvanishing. Then, the relevant quantities in $R_\mu v_{\ret}^\mu$ 
can be Taylor-expanded around $\tau_0$ (where the event $O'$ occurs) and then expand the position as $Z=Z_0+\delta$ with $\delta/Z\ll1$ such that, at this limit, the electric field component $E^Z$ is given by
\be
E^Z=\frac{4e}{g^2}\frac{1}{\left[\delta^2+2\delta Z_0+\tau_0^3Z_0\dot{\alpha}/3-\delta\tau_0^2Z_0\dot{\alpha}\right]^2}\ .
\en
In order to find a relation between both perturbative parameters $(\tau_0,\delta)$ we can use the relation for $R_\mu v_{\ret}^{\mu}$ and evaluate it at $\tau=0$ and $\rho=0$, finally obtaining that $\delta\sim\tau_0$. Then, the force over the particle exerted by its own field is simply $F^Z=eE^Z$ and expanded in powers of $\delta$ reads
\be \label{eq:radreact}
F^Z\approx\frac{e^2}{\delta^2}-\frac{e^2g}{\delta}+\frac{3}{4}g^2e^2+\frac{2}{3}e^2\dot{\alpha}+\mathcal{O}\paren{\delta}\ .
\en
The first two terms in Eq. \eqref{eq:radreact} are clearly divergent and are further absorbed through mass renormalization. The third term is constant and can be neglected due to the singular character of 
some terms in Eq. \eqref{eq:radreact}. Hence, the only relevant contribution from Eq. \eqref{eq:radreact} is the  term 
linear in $\dot{\alpha}$. This term can be directly transformed to the inertial system by transforming the force into the inertial coordinates, resulting directly in 
\be
f^z=\frac{2e^2}{3}\left(\dot{a}^{z}-v^z\left(\dot{a}^{\mu}v_{\mu}\right)\right)\ ,
\en
which is the radiation reaction term originally derived by Dirac.

\bibliography{Rad_BI.bib}

\end{document}